\documentclass[sigconf]{acmart}
\AtBeginDocument{%
  \providecommand\BibTeX{{%
    \normalfont B\kern-0.5em{\scshape i\kern-0.25em b}\kern-0.8em\TeX}}}

\usepackage[ruled,linesnumbered]{algorithm2e}

\copyrightyear{2021}
\acmYear{2021}
\setcopyright{rightsretained}
\acmConference[MobiHoc '21]{The Twenty-second International Symposium on Theory, Algorithmic Foundations, and Protocol Design for Mobile Networks and Mobile Computing}{July 26--29, 2021}{Shanghai, China}
\acmBooktitle{The Twenty-second International Symposium on Theory, Algorithmic Foundations, and Protocol Design for Mobile Networks and Mobile Computing (MobiHoc '21), July 26--29, 2021, Shanghai, China}\acmDOI{10.1145/3466772.3467053}
\acmISBN{978-1-4503-8558-9/21/07}

\newtheorem{theorem}{{\bf Theorem}}
\newtheorem{assumption}{{\bf Assumption}}

\newtheorem{lemma}{{\bf Lemma}}

\newtheorem{corollary}{{\bf Corollary}}

\DeclareMathOperator*{\argmin}{arg\,min}
\DeclareMathOperator*{\argmax}{arg\,max}

\begin{document}

\title{An Online Learning Approach to Optimizing Time-Varying Costs of AoI}

\author{Vishrant Tripathi}
\email{vishrant@mit.edu}
\affiliation{\institution{Massachusetts Institute of Technology} 
\country{USA}
}
%\email{vishrant@mit.edu}
%\author{Eytan Modiano}
%\email{modiano@mit.edu}
%\affiliation{%
%  \institution{Massachusetts Institute of Technology}
  %\streetaddress{P.O. Box 1212}
%  \country{USA}
\author{Eytan Modiano}
\email{modiano@mit.edu}
\affiliation{\institution{Massachusetts Institute of Technology} 
\country{USA}
}

\renewcommand{\shortauthors}{Tripathi and Modiano}

\begin{abstract}
We consider systems that require timely monitoring of sources over a communication network, where the cost of delayed information is unknown, time-varying and possibly adversarial.  
For the single source monitoring problem, we design algorithms that achieve sublinear regret compared to the best fixed policy in hindsight. %We further design algorithms that achieve sublinear regret compared to the best dynamic policy when the cost functions are slowly varying.
For the multiple source scheduling problem, we design a new online learning algorithm called \textit{Follow the Perturbed Whittle Leader} and show that it has low regret compared to the best fixed scheduling policy in hindsight, while remaining computationally feasible. The algorithm and its regret analysis are novel and of independent interest to the study of online restless multi-armed bandit problems. We further design algorithms that achieve sublinear regret compared to the best dynamic policy when the environment is slowly varying.
Finally, we apply our algorithms to a mobility tracking problem. We consider non-stationary and adversarial mobility models and illustrate the performance benefit of using our online learning algorithms compared to an oblivious scheduling policy. 
\end{abstract}

\begin{CCSXML}
<ccs2012>
   <concept>
       <concept_id>10003033.10003079.10003080</concept_id>
       <concept_desc>Networks~Network performance modeling</concept_desc>
       <concept_significance>500</concept_significance>
       </concept>
   <concept>
       <concept_id>10003033.10003079.10011672</concept_id>
       <concept_desc>Networks~Network performance analysis</concept_desc>
       <concept_significance>500</concept_significance>
       </concept>
   <concept>
       <concept_id>10003033.10003106.10010582.10011668</concept_id>
       <concept_desc>Networks~Mobile ad hoc networks</concept_desc>
       <concept_significance>500</concept_significance>
       </concept>
 </ccs2012>
\end{CCSXML}

\ccsdesc[500]{Networks~Network performance modeling}
\ccsdesc[500]{Networks~Network performance analysis}
\ccsdesc[500]{Networks~Mobile ad hoc networks}

\keywords{Age of Information, wireless networks, online learning, scheduling}

\maketitle

\section{Introduction}
\label{sec:intro}
%Monitoring/Estimation/Control over wireless networks is a well studied problem. Give examples, motivation, prior work (?).
%\subsection{Motivation}
%\label{sec:motivation}
Monitoring, estimation, and control of systems are fundamental  and well studied problems. Many emerging applications involve performing these tasks over communication networks. Examples include: networked control systems, sensing for IoT applications, control of robot swarms, real-time surveillance, and monitoring of sensor networks. In these settings, achieving good performance requires timely delivery of status updates from sources to destinations.

Age of Information (AoI) is a metric that captures timeliness of received information at a destination \cite{kaul2012real,kam2013age,yin17_tit_update_or_wait}. Unlike packet delay, AoI measures the lag in obtaining information at a destination node, and is therefore suited for applications involving time sensitive updates. Age of information, at a destination, is defined as the time that has elapsed since the last received information update was generated at the source. AoI, upon reception of a new update, drops to the time elapsed since generation of the update, and grows linearly otherwise. Over the past few years, there has been a rapidly growing body of work on analyzing AoI for queuing systems \cite{kaul2012real, kam2013age, yin17_tit_update_or_wait,huang2015optimizing,inoue2018general,bedewy2019minimizing}, using AoI as a metric for scheduling policies in networks \cite{kadota2018scheduling,kadota2018scheduling2,talak2018optimizing,tripathi2017age,farazi2018age,tripathi2019whittle} and for monitoring or controlling systems over networks \cite{sun2017remote,sun2019sampling,ornee2019sampling,champati2019performance,klugel2019aoi}. For detailed surveys of AoI literature see \cite{kosta2017age} and \cite{sun2019age_book}.

Typically, AoI is used as a metric for measuring freshness of information being delivered about a source to a monitoring station. It represents a measure of distortion between the state of the system that is expected at the monitor based on past updates and the actual current state of the system. Thus, a larger age corresponds to the monitor having a higher uncertainty about the current state of the system being observed. This, in turn, means that ensuring a low average AoI can lead to higher monitoring accuracy or better control performance. While AoI is a proxy for measuring the cost of having out-of-date information, it may not properly reflect the impact of  stale information on system performance.

When multiple systems or sources are being observed at the same time, there arises a need to differentiate between them based on their relative importance. So, many works on AoI-based scheduling for multiple sources consider weighted-sum AoI minimization \cite{kadota2018scheduling,kadota2018scheduling2,talak2018optimizing}, where weights represent the relative importance of each source. Typical assumptions involve the weights being fixed and known in advance, based on the underlying application or systems being monitored.
%Namely, if a system evolves slowly, then a higher AoI does not necessarily mean higher uncertainty or cost. Conversely, if a system has fast dynamics, then even a small AoI might mean higher distortion and poor performance. Further, certain sources might be more critical to the system performance than others. 

Further, recent works on networked control systems \cite{champati2019performance,klugel2019aoi} and remote estimation \cite{sun2017remote,sun2019sampling,ornee2019sampling} emphasize that even for very simple systems, linear AoI is not a sufficiently accurate metric to track accuracy or overall system performance. This has motivated interest in using general, possibly non-linear cost functions of AoI that reflect the cost of delayed information more accurately \cite{kosta2017nlage, jhun2018age, tripathi2019whittle,champati2019performance}. Typical assumptions in works studying non-linear AoI include knowing the cost functions in advance \cite{kosta2017age,jhun2018age,tripathi2019whittle,klugel2019aoi}, assuming that cost functions increase monotonically with AoI \cite{kosta2017age,tripathi2017age,sun2019sampling,klugel2019aoi} and decoupled costs across multiple systems \cite{tripathi2019whittle,klugel2019aoi}.%Consider the simple example of a discrete-time LTI system being estimated over a costly wireless channel. It was shown in \cite{champati2019performance} that this can be formulated as an AoI-based problem. However, the cost for AoI can increase possibly exponentially depending on the underlying system dynamics. 
%\subsection{Applications}
%\subsection{Prior Work}
%We have already discussed some works on AoI, remote estimation and networked control - how they relate to and motivate our work in Section \ref{sec:motivation}. Here, we discuss some closely related works that consider learning or adversarial behavior for remote estimation.

Learning how to sample a source through a network with an unknown delay profile while minimizing AoI has been considered in \cite{kam2019learning}. Minimizing AoI with unknown and adversarial channel processes has also been considered in \cite{bhandari2020age} and \cite{banerjee2020fundamental}, respectively. %In contrast, we consider general cost functions which are unknown and possibly time-varying.
%Remote estimation of a stationary random walk over a costly communication link is also considered in \cite{yun2018optimal}. Here the authors look at estimation error as the metric of interest. A threshold based policy is shown to be optimal, i.e. when the estimate of current estimation error goes above a threshold, a new sample is sent. The authors extend this setting to learning the optimal threshold for a stationary but unknown random walk. %Our model considers non-stationary sources and arbitrarily varying cost functions. 
%A remote estimation setting with an adversarial flavor is also considered in \cite{gac2018communication}. Here the authors consider a source to be monitored over a communication channel that adds \textit{adversarial noise} to every sample being sent. They formulate this as a nonzero-sum dynamic game and derive a Nash equilibrium. The focus of this work is on designing robust encoder-decoder pairs rather than sampling policies.

Importantly, we observe that all of these works assume there is some \textit{fixed and known} cost function mapping the AoI to system performance and that the source dynamics are stationary. In this work, we ask the question: \textit{what if this cost function is not known in advance, time-varying and possibly adversarial?} How does one go about designing scheduling policies that lead to good monitoring accuracy or control? Related to our work, a context-aware notion of AoI was proposed in \cite{zheng2019context}, where the authors considered sources with \textit{known} time-varying context that influences the AoI cost function.

Our goal is to model applications where delivering timely information is of essence but the costs for delayed information are not completely known beforehand and hard to model, including non-stationary settings and adversarial dynamics. A broad range of networked control and monitoring applications fit this description. An example is designing scheduling schemes for real-time monitoring of power grids which have nonlinear and complicated dynamics that cannot be easily modeled. Another example is scheduling for mobility tracking. Mobility traces in the real world are often highly non-stationary and hard to explain via models. A third example is monitoring queue length information in data centers for load balancing, where only a small number of queues are sampled every few time-steps, and traffic flows, server outages and job sizes may be non-stationary and possibly adversarial. All of the above problems require optimization of unknown time-varying cost functions of AoI in an online fashion. %We discuss the mo applications in greater detail in \ref{sec:applications}.

%We start with a single source version and formulate a problem that is amenable to applying techniques from online learning. We then extend our discussion to multiple sources and provide examples of applications. 
%\subsection{Outline}
In Section \ref{sec:single} we formulate a problem that involves monitoring a single non-stationary source over a costly communication channel. We design an epoch based framework in which the AoI cost functions change across epochs in an unknown time-varying manner, but remain fixed within an epoch. At the end of each epoch, the scheduler receives feedback (either partial or full) regarding the cost in the previous epoch and uses it to decide a policy for the next epoch\color{black}. We provide simple scheduling algorithms that have sublinear worst-case regret compared to the best fixed policy in hindsight. Our main contribution here is formulating the problem in such a way that we can apply techniques from online optimization. To the best of our knowledge, this is the first work to study monitoring and scheduling for non-stationary sources with unknown dynamics.%We do this for two kinds of feedback structure - full and bandit feedback. We then extend our results to incorporate notions of dynamic regret, and provide algorithms that also achieve sublinear dynamic regret, whenever possible.

In Section \ref{sec:multiple_sources}, we use insights from the singe source model to develop an epoch based framework for online scheduling of multiple sources. In each epoch, the scheduler needs to decide a scheduling policy that specifies which source gets to send an update in every time-slot. The goal is to dynamically adapt the scheduling policy to optimize for overall monitoring cost, as the AoI cost functions change across epochs in an unknown manner\color{black}. Since the number of scheduling policies of a given length grows exponentially in the number of sources, it becomes computationally infeasible to implement traditional online learning algorithms directly in the multiple source setting. We design a new online learning algorithm called \textit{Follow the Perturbed Whittle Leader} (FPWL) for this setting that is computationally feasible while also achieving low regret. Here, analyzing regret is especially challenging due to the combinatorial nature of the scheduling problem and since the Whittle index is only an approximately optimal solution for the offline problem. Our algorithm and its regret analysis are novel and of independent interest to the study of online learning for restless multi-armed bandits with time-varying costs.

In Section \ref{sec:applications}, we apply the algorithms that we develop to a mobility tracking problem and illustrate the performance benefits of using online learning for scheduling.
% and apply our results to sampling of queues for load balancing in data centers. We show via simulations that our online sampling methods can outperform power-of-d-choices based random sampling previously proposed in literature.

\section{Single Source Monitoring}
\label{sec:single}
%
%\subsection{Fixed Known Cost Function}
\begin{figure}
	\centering
	\includegraphics[width=0.8\linewidth]{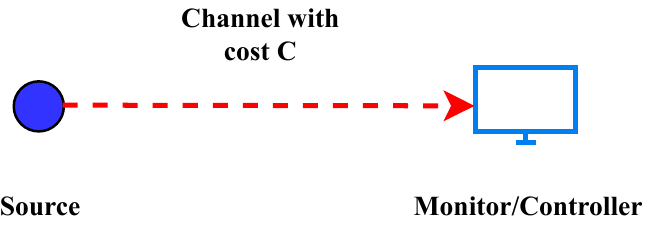}
	\caption{Single source monitoring}
	\label{figure:single}
\end{figure}
We start by discussing the single source setting with a known AoI cost function that remains fixed throughout. This will provide important insight and reveal key technical issues while formulating an online version of the problem. 

Consider a single source sending updates to a monitoring station over a costly wireless channel as in Figure \ref{figure:single}. In every time-slot, the scheduler decides whether the source sends a new update to the monitor. If it does, the monitor receives a new update in the next time-slot. The monitor maintains an age of information $A(t)$ which tracks how long it has been since it received a new update from the source. The evolution of $A(t)$ can be written as:
\begin{equation}
\label{eq:aoi_ev_single}
A(t+1) =
\begin{cases}
A(t)+1, & \text{if }u(t) = 0 \\
1, &  \text{if }u(t) = 1,
\end{cases}
\end{equation}
where $u(t)$ indicates whether a new update was sent in time-slot $t$. 

There is a known cost function of AoI $f(\cdot)$ which models the cost of having stale information at the monitor. Thus, the cost at any time-slot is:
\begin{equation}
	\text{Cost}(t) = f(A(t)) + C u(t),
\end{equation}
where $C > 0$ is the cost of sending a new update from the source. Our goal is to design a monitoring policy that minimizes the long term average cost. The average cost under a policy $\pi$ is: 
\begin{equation}
	\text{Cost}_{\text{ave.}}(\pi) \triangleq \limsup_{T \rightarrow \infty} \frac{1}{T} \sum_{t=1}^{T} f(A^{\pi}(t)) + C u^{\pi}(t). 
\end{equation}
This problem and the optimal policy have been analyzed in \cite{tripathi2019whittle} and \cite{klugel2019aoi}. We describe the key result from these works below.
%\begin{framed}
	\begin{lemma}
		\label{lem:th}
		The optimal policy for the single source monitoring problem is a stationary threshold policy. Let $H$ satisfy
		\begin{equation}
		f(H) \leq \frac{\sum_{j=1}^{H}f(j) + C}{H} \leq f(H+1).
		\label{eq:th}
		\end{equation}
		Then, the optimal policy is to send an update at time-slot $t$ only if $A(t) \geq H$. If no such $H$ exists, the optimal policy is to never send an update.
	\end{lemma}
%\end{framed}
\begin{proof}
	See Theorem 1 in \cite{tripathi2019whittle}.
\end{proof}
Lemma~\ref{lem:th} implies that the optimal monitoring policy is to send an update only if the current AoI gets above a threshold $H$. Similar threshold based schemes have appeared in many different settings for online sampling and remote estimation \cite{sun2017remote,sun2019sampling,ornee2019sampling,yun2018optimal}.

%\subsection{Simple Online Version}
Now consider a naive reformulation of the problem where the cost function $f(\cdot)$ is time-varying and unknown. We represent it by $f_t(\cdot)$ where the subscript indicates its time varying nature. If the function were fixed and unknown, we could have used reinforcement learning to solve the problem as done in \cite{kam2019learning}. However, this cannot be done directly for settings with time-varying costs.
%There are works in online learning that consider MDPs with time-varying rewards (see \cite{yu2009markov,even2009online}). However, they require strict conditions on mixing of the MDP under every possible policy. These are not satisfied by our formulation. Further, these methods are typically not computationally efficient in the settings of our interest (see Appendix \ref{} for a detailed discussion).  

On the other hand, there is a large amount of literature on online learning and optimization where the goal is to solve a sequence of optimization problems which vary in an unknown, possibly adversarial manner (see \cite{cesa2006prediction} and \cite{hazan2019introduction} for a detailed introduction to the field). In these problems there is no system state or history, so decisions in the current time-slot do not affect the optimization problem or decisions in a future time-slot. This is not true of our monitoring problem which has a system state (AoI), and where the state evolution depends on decisions taken in the past.

To overcome these difficulties, we reformulate the single source monitoring problem in an epoch based setting. %order to eliminate this history dependence. %Then, we apply results from online optimization to solve our problem.

\subsection{An Epoch Based Formulation}
\label{sec:model}
We observe that the AoI of the source resets to $1$ after every new update delivery. Thus, AoI evolution within an update inter-delivery period does not depend on the AoI evolution in any other period. We use this observation to formulate an epoch based problem.%Let $C$ represent a fixed communication cost.

We divide time into $T$ epochs, where each epoch further consists of $M$ time-slots. As before, when a new update is sent it gets delivered in the next time-slot. At the beginning of epoch $k$, we choose an AoI threshold $x_k \in \{1,...,M\}$. Within the epoch, the source generates a new sample and sends it to the monitor whenever the AoI reaches the threshold $x_k$. In the last time-slot of the epoch, a new sample is sent regardless of the AoI. This ensures that the next epoch begins with AoI at 1. A cost is observed for sending samples every $x_k$ time-slots based on the current system dynamics and communication costs. Then, epoch $k+1$ starts. Using cost information about previous decisions, a new sampling threshold $x_{k+1}$ is chosen for epoch $k+1$ and the process repeats itself. %Note that the length of each epoch is upper bounded by $M$ and thus the sampling rate is lower bounded by $1/M$.
%Epochs are thus time intervals over which our sampling policy remains fixed and at the end of which the system resets. At the beginning of each epoch a sampling policy is chosen based on the performance of the policies chosen in the past.

In each epoch $k$ there is a function $f_k(\cdot)$ that represents the current cost for age of information and remains fixed for the duration of the epoch. So, for any time-slot $t$ in epoch $k$, the current cost is given by $f_k(A(t)) + Cu(t)$. The total cost incurred in epoch $k$ denoted by $C_k(x)$ is simply the sum of the cost in the individual time-slots.  
\begin{lemma}
\label{lem:single_epoch}
If a sampling threshold of $x$ is chosen in epoch $k$ and the AoI cost function is $f_k(\cdot)$ then the loss function $C_k(x)$ is given by:
\begin{equation}
\label{eq:epoch_cost}
C_k(x) = \bigg\lfloor \frac{M}{x} \bigg\rfloor \bigg(\sum_{j=1}^{x}f_k(j) + C\bigg) + \vmathbb{1}_{r > 0} \bigg(\sum_{j=1}^{r}f_k(j) + C\bigg),
\end{equation}
where $r = M \mod x$. This is the sum total AoI cost of monitoring over the epoch $k$.
\end{lemma}
\begin{proof}
See Appendix \ref{pf:single_epoch}.
\end{proof}
If we knew $f_k(\cdot)$ at the beginning of epoch $k$, we could use \eqref{eq:epoch_cost} to find the optimal sampling threshold $x_k^{\ast}$. In our online framework, the goal is to learn the best sampling thresholds without knowing any information about the sequence of cost functions that we are going to face. 
%As noted earlier, the resetting of AoI at the beginning of each epoch means that the age evolution over epochs does not depend directly on decisions taken in the past epochs. Thus, we can view this formulation as a sequence of independent optimization problems over epochs where in each epoch we choose a sampling threshold and observe an associated cost. Assuming that the cost function remains fixed within an epoch but can change arbitrarily across epochs leads to an online learning problem.

While we have motivated the setting above using cost that splits into a sum of AoI cost and communication cost, our setup allows for general cost functions $C_k(x_k)$ that map the choice of sampling threshold $x_k$ to a cost in epoch $k$. For the remainder of this section, we will deal with these general cost functions $C_k(\cdot)$. 

In our online setting, an \textbf{unconstrained adversary} chooses the sequence of \textbf{bounded} cost functions $C_k(\cdot)$ for each epoch $k$. The designer does not have access to the sequence of cost functions beforehand and must learn a suitable transmission/sampling policy in an online manner. We make no assumptions on how the underlying system dynamics or resulting cost functions change across epochs.
%While epochs as defined above involve delivering a single update after waiting for some time $x$, they can be generalized to delivering $B$ updates with a sampling interval of $x$ in between each update, where $B$ is a fixed constant. The epochs should be roughly at the time-scale at which cost functions change and a new sampling interval needs to be chosen. So $B$ can be set beforehand based on the application. While all of our analysis will hold for any fixed value of $B$, for simplicity we will consider $B=1$ for the rest of the paper.

%\textit{Example:} 
Note that the cost function $C_k(\cdot)$ in epoch $k$ can be seen as an $M$ dimensional vector where the cost for choosing the sampling threshold $x$ is represented by the $C_k(x)$ which is the $x$th element of the vector. Going forward, when we use the notation $C_k$, we refer to the $M$ dimensional vector of costs for each threshold in epoch $t$, while $C_k(x)$ represents its $x$th element.

The boundedness of the cost functions $C_k$ is crucial to proving any meaningful results in this setting and is standard in online learning literature. Without loss of generality, we further assume that the cost functions are normalized such that $C_k(x) \in [0,1]$ for all sampling thresholds $x$ and epochs $k$.

\subsubsection{Feedback Structure} For the setup described above, we will look at two kinds of feedback structure for observing the costs. Note that $x_k$ represents the decision made at the beginning of epoch $k$.
\begin{itemize}
	\item {Full Feedback} - the scheduler observes the entire cost function $C_k(x), \forall x \in \{1,...,M\}$ at the end of epoch $k$.
	\item {Bandit Feedback} - the scheduler observes only $C_k(x_k)$ at the end of epoch $k$. 
	%\item {Partial Feedback} -  the scheduler observes $C_k(x), \forall x \in \{1,...,x_k\}$ at the end of epoch $k$.
\end{itemize}
%Full feedback is easy to analyze, however, it is not a practical assumption. 
\subsubsection{Objective (Regret Minimization)}: For any sequence of cost functions $C_1(\cdot), C_2(\cdot), ..., C_T(\cdot)$, $x^{*}$ is defined as the best fixed sampling threshold that minimizes sum AoI cost. It is given by the following equation.
\begin{equation}
\label{eq:opt_x}
x^{*} \triangleq \argmin \limits_{ x \in \{1,...,M\} } \sum_{k=1}^{T} C_k(x).
\end{equation}
Our goal is to find an online policy that achieves sublinear regret compared to the best fixed sampling threshold $x^{*}$ for any sequence. This is known as \textit{worst-case static regret}. For any policy $\pi$, it is defined as follows:
	\begin{equation}
	\label{eq:regret}
	\text{Regret}_T(\pi) = \sup_{C_1,...,C_T} \bigg\{  \sum_{k=1}^{T} C_k(x^{\pi}_k) - \min_{ x \in \{1,...,M\} } \sum_{k=1}^{T} C_k(x) \bigg\}.
	\end{equation}	
%\end{definition}
Note that regret is defined over epochs rather than time-slots since we assume that cost functions can change only across epochs. %We will discuss how to achieve sublinear static regret in Section \ref{sec:static_regret}.

We will now show that our online sampling problem formulation is equivalent to the prediction with expert advice setting that is well studied in online learning literature. This will allow us to apply policies and regret bounds derived for this setting to our problem.

\subsubsection{Prediction With Expert Advice:} A decision maker has to choose among the advice of $n$ given experts. After making a choice, a bounded loss is incurred. This scenario is repeated iteratively, and at each iteration the costs of choosing the various experts are arbitrary (possibly even adversarial, trying to mislead the decision maker). The goal of the decision maker is to do as well as the best expert in hindsight.

In our setting, the role of experts is played by the AoI thresholds $x \in  \{1,...,M\}$. In each epoch, the scheduler decides an AoI sampling threshold $x$ and observes an associated cost. This process repeats iteratively with time-varying, possibly adversarial changes to costs. Thus, our setting corresponds with the expert advice setting with $M$ experts. %Our goal is to find the best fixed sampling threshold in hindsight.

\subsubsection{Sublinear Regret}
We now discuss in detail a policy that achieves sublinear static regret for the full feedback setting. This will illustrate how regret bounds from online learning literature can be applied to our single source online monitoring setup.  

%\subsection{Static Regret}
%\label{sec:static_regret}
%\subsubsection{Full Feedback}
%\label{sec:single_full}
The full feedback assumption in our setting means that we observe costs for all possible sampling thresholds in every epoch. This makes sense when the scheduler has information about the current source dynamics and communication costs by the end of an epoch. Knowing this information is often sufficient to construct the current cost function for any possible sampling threshold. %Since the entire cost function is available at the end of each epoch in this setting, it is relatively easy to achieve low regret.

We describe an online monitoring policy based on Follow the Perturbed Leader (FTPL) style algorithms. The FTPL method was first analyzed in the online setting in \cite{kalai2005efficient} and is based on an algorithm first proposed in \cite{hannan1957approximation}. The key idea of the FTPL algorithm is to maintain the sum of cost functions observed until now and perturb it slightly. Choosing the best AoI threshold based on this perturbed history is sufficient to get sublinear regret. %Algorithm \ref{alg:FTPL} describes the details of the policy. %The past costs are accumulated in an $M$ dimensional vector $\Theta_t$. The noise perturbation is also $M$ dimensional and i.i.d. Gaussian. As before, the $x$th element of a vector $V$ is denoted by $V(x)$.
\begin{algorithm}
	\DontPrintSemicolon
	%\SetAlgoLined
	%\KwResult{Write here the result}
	\SetKwInOut{Input}{Input}\SetKwInOut{Output}{Output}
	\Input{parameter $\eta > 0$, number of thresholds $M$}
	%\Output{Write here the output}
	\BlankLine
	Set $\Theta_1 \leftarrow 0$ \\		
	\While{ $t \in 1,...,T$ }{
		Sample $\gamma_t \sim \mathcal{N}(0,I)$\;
		Choose sampling threshold $x_t \in \argmin\limits_{ x \in \{1,...,M\} } \Theta_t(x) + \eta \gamma_t(x)$\;
		Incur loss $C_t(x_t)$ and update $\Theta_{t+1} = \Theta_{t} + C_t$ 				
	}		
	\caption{FTPL for Online Monitoring}
	\label{alg:FTPL}
\end{algorithm}
%The regret bounds for an FTPL policy are for expected regret where the expectation is taken over the random perturbations $\gamma_t$. By choosing the parameter $\eta$ carefully, we can show that Algorithm \ref{alg:FTPL} achieves $O(\sqrt{T})$ regret. Theorem \ref{thm:FTPL_regret} provides an upper bound for the expected regret.

%\begin{framed}
	\begin{theorem}
		FTPL online monitoring described by Algorithm \ref{alg:FTPL} with $\eta = \sqrt{T}$ achieves the following upper bound for expected regret:
		\begin{equation*}
			\mathbb{E}[\text{Regret}_T(\text{FTPL})] \leq 2 \sqrt{2T \log M },
		\end{equation*}
		where the expectation is taken over the random perturbations.
		\label{thm:FTPL_regret}
	\end{theorem}
%\end{framed}

\begin{proof}
The proof is based on \cite{cohen2015following}. A lower bound of the form $\Omega(\sqrt{T \log M})$ is also available in \cite{cesa2006prediction}. %Details are included in the technical report.%Together, these results show that the FTPL policy is near optimal.
\end{proof}

\subsubsection{Bandit Feedback}
The bandit feedback assumption implies that we only observe the cost associated with the chosen sampling threshold. This is a realistic assumption especially when no other information about the system dynamics and communication costs is available to the scheduler. However, the single point feedback means learning happens slowly and regret bounds are worse in this setting.

The online bandit setting has also been well studied in literature. Notably, the EXP3 algorithm was first proposed in the seminal paper \cite{auer2002nonstochastic} and is known to have near optimal expected regret under bandit feedback. We describe online monitoring based on EXP3 below.
\begin{algorithm}
	\DontPrintSemicolon
	%\SetAlgoLined
	%\KwResult{Write here the result}
	\SetKwInOut{Input}{Input}\SetKwInOut{Output}{Output}
	\Input{parameter $\epsilon > 0$, distribution $p_1 = \mathbf{1}/M$}
	%\Output{Write here the output}
	%\BlankLine		
	\While{ $t \in 1,...,T$ }{
		Choose sampling threshold $x_t \sim p_t$\;
		Incur loss $C_t(x_t)$ and observe  $C_t(x_t)$\;
		Let 
		\begin{equation*}
		\hat{C}_t(i) =
		\begin{cases}
		C_t(i)/p_t(i), & \text{ if }i = x_t  \\
		0, &  \text{ otherwise.}
		\end{cases}
		\end{equation*}\;
		Update $y_{t+1}(i) = p_t(i)e^{-\epsilon \hat{C_t}(i)}$, $p_{t+1}= \frac{y_{t+1}}{||y_{t+1}||_1}$				
	}		
	\caption{EXP3 for Online Monitoring}
	\label{alg:exp3}
\end{algorithm}

The key idea of Algorithm \ref{alg:exp3} is to maintain an unbiased estimate of the cost $C_t$ via importance sampling (line 4). In every epoch, the algorithm samples a threshold $x_t$ from a probability distribution $p_t$ over the $M$ thresholds. At the end of the epoch, $p_t$ is updated with the current cost function estimate using exponential weights. It can be shown that the expected regret of EXP3 is sublinear and near optimal. Theorem \ref{thm:EXP3_regret} provides an upper bound on the expected regret in the bandit feedback setting.

%\begin{framed}
	\begin{theorem}
		The EXP3 online sampling policy described by Algorithm \ref{alg:exp3} with $\epsilon = \sqrt{\frac{\log M}{T M}}$ achieves the following upper bound for expected regret:
		\begin{equation*}
		\mathbb{E}[\text{Regret}_T(\text{EXP3})] \leq 2 \sqrt{T M \log M }.
		\end{equation*}
		The expectation is taken over the random sampling decisions made in each epoch.
		\label{thm:EXP3_regret}
	\end{theorem}
%\end{framed}

\begin{proof}
The proof and a lower bound of the form $\Omega(\sqrt{T M})$ follow from discussion in \cite{auer2002nonstochastic}.% Details are included in the technical report.%Regret in the bandit feedback setting can be lower bounded by  \cite{auer2002nonstochastic}.%For details, see Appendix \ref{}.
\end{proof}

%Just as for the full feedback case, algorithms that achieve sublinear static regret in the bandit feedback case have been well studied for prediction with experts (see EXP3 proposed in \cite{auer2002nonstochastic}). We can directly apply them to our online monitoring formulation to design sampling policies that achieve sublinear regret. \color{blue}We include the a review of these methods and their application to our setting in our technical report.\color{black}

Next, we discuss what sublinear epoch regret means for AoI cost averaged over time-slots. To do so, we note that the sum total AoI cost over all time-slots equals the cost summed over individual epochs, by definition. Let $\pi$ denote an online algorithm which specifies threshold $x^{\pi}_k$ to be chosen in epoch $k$ and let $E_k$ be the set of times-slots in epoch $k$. Then,
\begin{equation}
\label{eq:epoch_time}
    \sum_{k=1}^{T} C_k(x^{\pi}_k) = \sum_{k=1}^{T} \sum_{t \in E_k} f_k(A^{\pi}(t)) + C u^{\pi}(t).  
\end{equation}
The relation above immediately implies that epoch regret also equals regret over time-slots. We describe this in the lemma below.
\begin{lemma}
\label{lem:epoch_time}
    Suppose an online algorithm $\pi$ has an upper bound on its expected static epoch regret of the form $f(M,T)$. Let $\pi^{*}$ denote the policy corresponding to the best fixed AoI threshold $x^{*}$ given the entire sequence of AoI cost functions $f_1,...,f_T$ and the sampling cost $C$. Then for any bounded sequence of cost functions, the same upper bound holds for regret over time-slots, i.e.
    \begin{equation}
    \begin{split}
        \mathbb{E}\Bigg[\sup_{f_1,...,f_T} \bigg\{  \sum_{k=1}^{T} \sum_{t \in E_k} f_k(A^{\pi}(t)) + C u^{\pi}(t) -\\ \sum_{k=1}^{T} \sum_{t \in E_k} f_k(A^{\pi^{*}}(t)) + C u^{\pi^{*}}(t) \bigg\}\Bigg] \leq f(M,T),
    \end{split}
    \end{equation}

    % \begin{equation}
    % \begin{split}
    %     \limsup_{T \rightarrow \infty} \frac{1}{MT} \mathbb{E}\bigg[ \sum_{k=1}^{T}\sum_{t=kM-M+1}^{kM} f_k(A^{\pi}(t)) + C u^{\pi}(t)  \bigg] \leq \\ \limsup_{T \rightarrow \infty} \frac{1}{MT} \mathbb{E}\bigg[ \sum_{k=1}^{T}\sum_{t=kM-M+1}^{kM} f_k(A^{\pi^{*}}(t)) + C u^{\pi^{*}}(t)  \bigg].
    % \end{split}
    % \end{equation}
    %Here $A^{\pi}(t)$ is the AoI at the monitor which evolves according to \eqref{eq:aoi_ev_single} under the sampling decisions made by $\pi$.% and $f_k(\cdot)$ is the AoI cost function in epoch $k$.
\end{lemma}

\begin{proof}
Substituting $\sum_{k=1}^{T}C_k(x^{\pi}_k)$ in \eqref{eq:regret} using \eqref{eq:epoch_time} gives us the required result.
\end{proof}

%Let the AoI under an online sampling policy $\pi$ at time $t$ be given by $A^{\pi}(t)$. We denote the best fixed policy in hindsight by $\pi^{*}$. This policy chooses the optimal sampling threshold $x^{*}$, given by \eqref{eq:opt_x}, in every epoch. Let the AoI under $\pi^{*}$ at time $t$ be denote by $A^{\pi^{*}}(t)$. Similarly, let $u^{\pi}(t)$ and $u^{\pi^{*}}(t)$ be indicator variables that are 1 when policies $\pi$ and $\pi^{*}$ respectively decide to send an update. We define time-slot regret as:

%long run AoI averaged over time-slots. Any online algorithm with sublinear static regret achieves a total time-average AoI cost that is less than or equal to the cost under the best fixed sampling threshold. We highlight this in the following lemma. %Similarly, any policy with sublinear dynamic regret achieves a time average AoI cost that is less than or equal to the cost of the policy that chooses the best sampling threshold in each epoch. 
If $f(M,T)$ is sublinear in the number of epochs $T$, then it is also sublinear in the number of time-slots $MT$ since we assume that $M$ is fixed to be a large constant. Thus, using Lemma \ref{lem:epoch_time}, sublinear epoch regret implies sublinear time-slot regret. As a direct corollary of this, any online algorithm with sublinear static epoch regret achieves an expected time-average AoI cost which is at least as good as that under the best fixed sampling threshold.
\begin{corollary}
\label{corr:1}
    Suppose an online algorithm $\pi$ has an upper bound on its expected static regret that grows sublinearly in $T$. Let $\pi^{*}$ denote the policy corresponding to the best fixed AoI threshold $x^{*}$ given the entire sequence of AoI cost functions $f_1,...,f_T$. Then for any sequence of bounded cost functions the following holds:
    \begin{equation}
    \begin{split}
        \limsup_{T \rightarrow \infty} \frac{1}{MT} \mathbb{E}\bigg[ \sum_{k=1}^{T}\sum_{t \in E_k} f_k(A^{\pi}(t)) + C u^{\pi}(t)  \bigg] \leq \\ \limsup_{T \rightarrow \infty} \frac{1}{MT} \mathbb{E}\bigg[ \sum_{k=1}^{T}\sum_{t \in E_k} f_k(A^{\pi^{*}}(t)) + C u^{\pi^{*}}(t)  \bigg].
    \end{split}
    \end{equation}
    %Here $A^{\pi}(t)$ is the AoI at the monitor which evolves according to \eqref{eq:aoi_ev_single} under the sampling decisions made by $\pi$.% and $f_k(\cdot)$ is the AoI cost function in epoch $k$.
\end{corollary}
\begin{proof}
See Appendix \ref{pf:corr}.
\end{proof}
Note that the relation in Corollary 1 is an inequality and not an equality because we are comparing to the best static threshold policy across epochs and it is possible that our online monitoring policy performs better. %The result above guarantees that it cannot perform any worse in a time-average sense.% since it is able to change thresholds across epochs.
\section{Multiple Sources}
\label{sec:multiple_sources}
\begin{figure}
	\centering
	\includegraphics[width=0.85\linewidth]{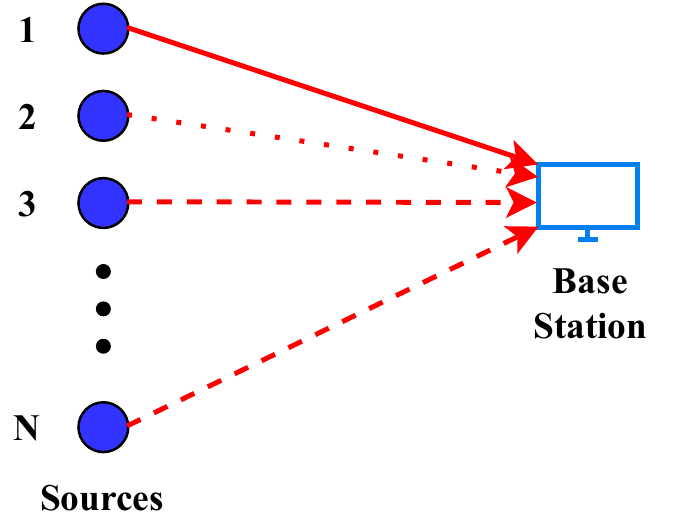}
	\caption{Multiple source monitoring}
	\label{figure:multiple}
\end{figure}
Motivated by the single-source discussion, we study a more challenging problem. Now, multiple sources are sending information to a monitoring station over a network as in Figure \ref{figure:multiple}. In this setting, the scheduler needs to decide which source gets to send an update in every time-slot to optimize for overall monitoring accuracy and performance, and the goal is to learn good scheduling policies.

Consider a system with $N$ sources sending updates over a network such that only one source can transmit at any given time-slot (due to interference/capacity constraints). We assume reliable channels, i.e. when a source is chosen to transmit an update, it is delivered to the monitor without fail in the next time-slot. Freshness aware scheduling in such single-hop wireless networks has been the focus of a lot of recent work in the AoI community \cite{kadota2018scheduling,kadota2018scheduling2,talak2018optimizing,tripathi2017age,farazi2018age,tripathi2019whittle,jhun2018age}. %It was established in \cite{jhun2018age} that the optimal scheduling policy for fixed cost functions of AoI (with reliable channels) is cyclic, meaning that it repeats a finite sequence of scheduling decisions in a periodic manner. 

We now create an epoch-based structure and set up an online learning formulation for multiple sources as we did in the single source setting. As before, we divide the time horizon into $T$ epochs, where each epoch is of length $M$ time-slots. At the beginning of epoch $k$, the scheduler needs to decide a scheduling policy $\pi_k$ which specifies when to schedule each sensor. Once the epoch is over, a cost of the form $C_k(\pi_k)$ is incurred and a new epoch begins. Using cost information about previous decisions, we again choose a scheduling policy $\pi_{k+1}$ for epoch $k+1$ and the process repeats itself. %The length of each epoch is fixed to be a constant $M>N$.

We maintain variables $A^{(1)},...,A^{(N)}$ which track the evolution of AoI for each source within an epoch. The evolution of AoI for source $i$ in epoch $k$ is described by the following equation:	
\begin{equation}
A^{(i)}(j+1) =
\begin{cases}
A^{(i)}(j)+1, & \text{if } i \notin \pi_k(j)  \\
1, &  \text{if } i \in \pi_k(j) ,
\end{cases}
\label{eq:aoi_ev}
\end{equation}
where $j$ is an index denoting the current time-slot within the epoch and $\pi_k(j)$ is the scheduling decision set in time-slot $j$ of epoch $k$. %Since the epoch length is fixed to $M$, AoI during an epoch can go from 1 to $M$. 

Similar to the single-source formulation, we relax the interference constraint in the last time-slot of every epoch. This ensures that the AoI of every source is set to 1 at the beginning of each epoch and we do not need to maintain history of AoI across different epochs. Practically, we justify this as a two time-scale assumption. A scheduling policy remains fixed over an epoch (the longer time-scale) and specifies how to take decisions over time-slots (the shorter time-scale). Once the epoch ends, the system resets. The system designer observes the performance of the scheduling policy that was chosen and specifies a new scheduling policy for the next epoch. 

We consider scheduling policies as a sequence of $M$ scheduling decisions, specifying which source gets to transmit in each time-slot within an epoch. We denote this space of scheduling policies by $\Pi^{M}$. %We consider $M$ to be a large so that $\Pi^{M}$ includes a wide range of scheduling policies.

We assume that the cost for delayed information in any time-slot can be represented as a general function of the current AoIs.  Let $f_k(A^{(1)},...,A^{(N)})$ represent this AoI cost function in epoch $k$, where $f_k:{\mathbb{Z}^{+}}^{N} \rightarrow [0,D]$ is a \textbf{bounded} mapping from the set of AoI vectors to costs. The total cost of choosing a policy $\pi$ in epoch $k$ is given by
\begin{equation}
\label{eq:finite_cost_g}
\sum\limits_{j=1}^{M}  f_k(A^{(1)}(j),...,A^{(N)}(j)),
\end{equation}
where the AoIs evolve under policy $\pi$ according to \eqref{eq:aoi_ev}.

We have an \textbf{unconstrained adversary} who chooses the sequence of bounded cost functions $f_k(\cdot)$ for each epoch $k$ and we need to learn the best scheduling policy in response to any sequence of cost functions. Without loss of generality we assume that $f_k(\cdot)$ are normalized such that the total cost of any policy $C_k(\cdot)$, given by \eqref{eq:finite_cost_g}, lies in the set $[0,1]$. %While we have motivated the setting above using epoch costs that are a long-run time-average of AoI cost functions, our setup allows for more general costs $C_k(\pi_k)$ that map the choice of scheduling policy $\pi_k$ to a cost in $[0,1]$. For the following discussion, we will deal with these general cost functions $C_k:\Pi^M \rightarrow [0,1]$.  

%The resetting of AoIs at the beginning of each epoch implies that scheduling decisions in one epoch do not affect future scheduling decisions, meaning there is no system state or history that needs to be maintained across epochs. %Note that this assumption is harder to justify in the multiple source setting since the AoIs of all sources do not reset to $1$ at the beginning of a new epoch. However, if we account for system state (the source AoIs), then the problem can no longer be formulated within the online learning framework. So, we make the relaxing assumption above.
%While epochs as defined above involve a single scheduling cycle of $s_k$ in the $k$th epoch, they can be generalized to $B$ cycles that repeat the schedule $s_k$, where $B$ is a fixed constant. The epochs should be roughly at the time-scale at which cost functions change and a new scheduling policy needs to be chosen. So $B$ can be set beforehand based on the application. As before, all of our analysis holds for any fixed value of $B$, but for simplicity we will consider $B=1$.

At the end of every epoch, the scheduler receives feedback in terms of $C_k(\cdot)$. In the case of full feedback, the entire function $C_k(\cdot)$ is revealed, meaning cost for all scheduling policies is known when the epoch ends. For the case of bandit feedback, only the cost for the chosen scheduling policy $C_k(\pi_k)$ is revealed.

Observe that the multiple source problem with the feedback structure as set up above can also be viewed as prediction with expert advice. Now, instead of AoI thresholds as experts, we have scheduling policies as experts and our goal is to \textit{compete with the best scheduling policy in hindsight.} 

Thus, we can directly apply online learning algorithms as done in Section \ref{sec:single} to the multiple source setting. The regret bounds, however, are not the same. This is because the number of scheduling policies of length $M$ time-slots scales as $\Theta(N^{M})$. %Thus, both static and dynamic regret bounds in the multiple source setting can be obtained by replacing $M$ with $N^{M}$ in the corresponding bound for the single source setting.
\begin{lemma}
		Consider the multiple source online scheduling problem with $N \geq 2$. If an online algorithm $\text{Alg.}$ has an upper bound $f(M,T)$ on its expected regret in the single source setting, then the same algorithm run using scheduling policies as experts for the multiple source problem has the following regret bound:
		\begin{equation*}
		\mathbb{E}[\text{Regret}_T(\text{Alg.})] \leq \bar{C}f(N^{M},T),
		\end{equation*}
		where $\bar{C}>0$ is a constant that does not depend on any other parameters. %This holds for both static and dynamic regret.
		\label{lem:multiple}
\end{lemma}

We observe that while the dependence of regret on $T$ remains the same, it becomes exponentially worse in $M$ for the multiple source setting. This also highlights a key computational challenge in the multiple source setting. The number of policies scales exponentially with $M$, the length of an epoch. Thus the optimization step in FTPL (Algorithm \ref{alg:FTPL}) has computational complexity that scales exponentially with $M$. Similar computational challenges are faced in implementing exponential weight algorithms like EXP3 for the bandit feedback case of the multiple source setting. This makes it hard to implement these online scheduling schemes in practice. %Similar observations have been made in online learning works when attempting to solve online optimization problems with combinatorial structure \cite{neu2013efficient,audibert2014regret}.

This is not surprising, given that the offline problem of finding the best scheduling policy of length $M$ time-slots in the setting with cost functions known beforehand also requires computation that scales as $O(N^{M})$ (see \cite{jhun2018age}). In \cite{tripathi2019whittle}, the authors analyzed the setting where cost functions can be represented as sums of separate cost functions that depend only on the AoI of each source individually. If the individual cost functions of AoI for each source are monotone increasing, then a low complexity heuristic based on the Whittle index approach can be found which is nearly optimal. We will use this observation to design low complexity online policies that keep track of the best scheduling policy in hindsight.
\subsection{Online Whittle-Index Scheduling}
We modify the general multiple source setting so as to solve the computational challenge discussed above. 
%First, we maintain variables $A^{(1)},...,A^{(N)}$ which track the evolution of source AoIs within an epoch. At the beginning of every epoch, the AoIs of all sources are set to $1$, to ensure that the there is no history dependence or state that needs to be maintained across epochs. The evolution of AoI is described by the following equation:	
%\begin{equation}
%	A^{(i)}(j+1) =
%	\begin{cases}
%	A^{(i)}(j)+1, & \text{if } i \notin \pi_k(j)  \\
%	1, &  \text{if } i \in \pi_k(j) ,
%	\end{cases}
%\end{equation}
%where $j$ is an index denoting the current time-slot within the epoch and $\pi_k(j)$ is the scheduling decision set in time-slot $j$ of epoch $k$. Since epoch length is upper bounded by $M$, AoI during an epoch can go from 1 to $M$. 

First, we consider scheduling policies as mappings from the set of AoI vectors $A^{(1)},...,A^{(N)}$ to the set of sources, i.e. $\pi:{\mathbb{Z}^{+}}^{N} \rightarrow \{1,..N\}$. Given the AoIs of all sources at time-slot $j$ within an epoch, a policy $\pi$ specifies which source gets to transmit. We denote this space of scheduling policies by $\Pi$. %Epoch lengths are fixed to be $M$, as before.
%If the policy $\pi$ is cyclic, as is the case with index policies \cite{tripathi2019whittle}, we fix the epoch length to be a single cycle of the policy starting from all AoIs set to $1$. If the length of the cycle turns out to be greater than $M$ or if the policy is not cyclic, we truncate it and end the epoch after $M$ time-slots, where $M$ is chosen to be a large constant.

Second, we assume that the cost function splits as a sum of individual cost functions of AoI, where $f_k^{(1)},...,f_k^{(N)}$ represent individual AoI cost functions for each source in epoch $k$. Then, the total cost of choosing a policy $\pi$ in epoch $k$ is given by
\begin{equation}
	\label{eq:split_cost}
	C_k(\pi) = \frac{1}{NM}\sum\limits_{j=1}^{M} \sum\limits_{i=1}^{N} f_k^{(i)}(A^{(i)}(j)),
\end{equation}
where the AoIs evolve under policy $\pi$ according to \eqref{eq:aoi_ev}. We multiply a normalizing constant $\frac{1}{NM}$ to the sum AoI cost to make regret analysis neater. 
%This is the long-run time-average cost if policy $\pi$ was used indefinitely and the cost functions $f_k^{(1)},...,f_k^{(N)}$ were the same throughout.

We assume that the cost functions $f_k^{(i)}:\mathbb{Z}^{+}\rightarrow\mathbb{R}^{+}$ are fixed during an epoch and bounded monotone increasing functions of AoI, i.e. if $x > y$ then $f_k^{(i)}(x) \geq f_k^{(i)}(y)$ and $f_k^{(i)}(\cdot) \leq D$. An unconstrained adversary is free to change these bounded cost functions arbitrarily across epochs.

Finally, instead of receiving feedback directly in terms of cost of scheduling policies $C_k:\Pi^{M}\rightarrow[0,1]$, we consider feedback in terms of individual cost functions of AoI. So, at the end of epoch $k$, a cost $C_k(\pi)$ is incurred (given by \eqref{eq:split_cost}) and AoI cost functions $f_k^{(1)},...,f_k^{(N)}$ are revealed to the scheduler, either completely or partially. In the case of bandit feedback, we will construct estimates of the entire cost functions $\hat{f}_k^{(1)},...,\hat{f}_k^{(N)}$.

%Note that $C_k(\pi_k)$ is the long-run time-average cost if the epoch had continued indefinitely with the chosen policy $\pi$, as given by \eqref{eq:long_run_cost}. 
Note that within an epoch, the scheduling problem that we want to solve is an instantiation of the functions of age problem described in \cite{tripathi2019whittle}.% assuming that the epochs are long.

We briefly review the multiple source setting with fixed AoI cost functions studied in \cite{tripathi2019whittle}. Consider $N$ sources and a given set of increasing AoI cost functions $f^{(1)},...,f^{(N)}$. Our goal is to minimize average age cost over an infinite horizon. The Whittle index policy maps the current vector of source AoIs to a scheduling decision. If the current AoI for source $i$ is $A^{(i)}$ then the Whittle policy is given by
\begin{equation}
\pi^{\text{Whittle}}(A^{(1)},...,A^{(N)}) \triangleq \argmax_{i \in \{1,...,N\}} \{ W^{(i)}(A^{(i)}) \},
\label{eq:whittle_policy}
\end{equation} 
where 
\begin{equation*}
	W^{(i)}(x) \triangleq x f^{(i)}(x+1) - \sum_{k=1}^{x} f^{(i)}(k)	
\end{equation*} are Whittle index functions. It was shown in \cite{tripathi2019whittle} that this Whittle policy is optimal for $N=2$ and near optimal in general. For cost functions $f^{(1)},...,f^{(N)}$, we denote the Whittle policy given by \eqref{eq:whittle_policy} as $\text{Whittle} \big(f^{(1)},...,f^{(N)}\big)$. Next, we describe how to design a low-complexity online algorithm using Whittle index policies.

\subsubsection{Full Feedback}
\label{sec:multiple_full}
In this setting, we assume that the entire $M$ dimensional AoI cost function $f_k^{(i)}$ for each source $i$ is revealed to the scheduler at the end of the epoch. Instead of looking for the best schedule in every epoch which is computationally expensive, we will use the Whittle index policy as an approximate minimizer. This leads to Algorithm \ref{alg:FTWL}, which we call \textit{Follow the Perturbed Whittle Leader} (FPWL).

\begin{algorithm}
	\DontPrintSemicolon
	%\SetAlgoLined
	%\KwResult{Write here the result}
	\SetKwInOut{Input}{Input}\SetKwInOut{Output}{Output}
	\Input{parameter $\epsilon > 0$}
	%\Output{Write here the output}
	\BlankLine
	Set $F_1^{(i)}(j) = j, \forall i \in \{1,...,N\}, \forall j \in \{1,...,M\}$ \\		
	\While{ $t \in 1,...,T$ }{
		Set $A^{(1)},...,A^{(N)} = \mathbf{1}$\;
		Sample $\delta_t^{(i)}(j) \sim \text{ uniform in } [0,1/\epsilon], \text{ i.i.d. }\forall i \in \{1,...,N\} \text{ and } \forall j \in \{1,...,M\} $\;
		Compute $\gamma_t^{(i)}(j) = \sum_{k=1}^{j} \delta_t^{(i)}(j), \forall i,j$ \;
		Choose scheduling policy $\pi_t = \text{Whittle} \bigg(F_t^{(1)} + \gamma_t^{(1)},...,F_t^{(N)} + \gamma_t^{(N)}\bigg)$\;
		Incur loss = $C_t(\pi_t)$ over epoch $t$ and observe feedback on $f_t^{(1)},...,f_t^{(N)}$\;
		In case of bandit feedback, construct cost estimates $\hat{f}_t^{(i)}, \forall i \in \{1,...,N\}$ using linear interpolation\;
		Update
		\begin{equation*}
		F_{t+1}^{(i)} =
		\begin{cases}
		F_{t}^{(i)} + f_t^{(i)}, \forall i \in \{1,...,N\}, \text{ if full feedback }  \\
		F_{t}^{(i)} + \hat{f}_t^{(i)}, \forall i \in \{1,...,N\}, \text{ if bandit feedback.}
		\end{cases}
		\end{equation*}				
	}		
	\caption{Follow the Perturbed Whittle Leader}
	\label{alg:FTWL}
\end{algorithm}

FPWL can be divided into three major steps. First, accumulate the entire history of cost functions that the scheduler has seen until the current epoch in $F_t^{(1)},...,F_t^{(N)}$. Since cost functions in each epoch are increasing in terms of AoI, their sums $F_t^{(i)}$ are also increasing. Second, perturb these accumulated cost functions in a manner such that they remain increasing functions of AoI but are still amenable for FTPL style analysis. Third, instead of computing the best possible scheduling policy for these accumulated and perturbed cost functions, use the Whittle index policy as an approximate minimizer.
%Note that we have replaced the step involving minimization over all scheduling policies with finding the best Whittle policy given the history of cost functions seen up to the current epoch. 

Computing the Whittle policy has complexity $O(NM)$ since it involves a maximization over $N$ quantities for at most $M$ steps. Further, generating the random perturbations $\gamma_t$ in steps 4 and 5 also takes at most $O(NM)$ computation. Thus, the algorithm above resolves the computational challenge involved in implementing FTPL directly for the online scheduling problem.
%It was shown in \cite{tripathi2019whittle} that Whittle policies are also cyclic. So the $\text{Whittle}(\cdot)$ operation in line 4 of the algorithm above computes a single cycle of the Whittle policy given cost function $F_t^{(i)} + \eta \gamma_t^{(i)}$ for source $i$ and starting from AoI of all sources set to 1. If the cycle turns out to be longer than $M$ time-slots, we simply truncate the remaining part.

Proving regret bounds our proposed algorithm is much harder than in the single or multiple source settings studied earlier. We overcome three significant  problems: 1) perturbations in Algorithm \ref{alg:FTWL} are made to the AoI cost functions rather than policies, unlike regular FTPL; 2) because of this, cost perturbations are not i.i.d. across policies; 3) the Whittle index policy is only an approximate minimizer rather than an exact minimizer of the average AoI cost. Despite these challenges, we are able to show that FPWL achieves low regret compared to any fixed scheduling policy, if the Whittle policy is ``close" to the actual optimal policy. Theorem \ref{thm:regret_multi} describes an upper bound on the regret of FPWL when compared to the best fixed scheduling policy in hindsight. The parameter $\alpha$ measures the closeness between the Whittle policy and an optimal policy. For a detailed definition of $\alpha$ see Appendix \ref{pf:alpha}.
\begin{theorem}
		Follow the perturbed Whittle leader (FPWL) based scheduling described by Algorithm \ref{alg:FTWL} with $\epsilon = \sqrt{\frac{2M}{N D^2 T}}$ achieves the following upper bound on expected regret:
		\begin{equation*}
		\mathbb{E}[\text{Regret}_T(\text{FPWL})] \leq \alpha T + 2D\sqrt{2MNT},
		\end{equation*}
		where the expectation is taken over the random perturbations.
		\label{thm:regret_multi}
\end{theorem}
\begin{proof}
See Appendix \ref{pf:w_regret}.
\end{proof}

It was proved in \cite{tripathi2019whittle} that the Whittle index policy is optimal for $N=2$, meaning $\alpha=0$ and we can achieve sublinear regret with respect to the best fixed scheduling policy when there are $2$ sources. Further, recent work in \cite{maatouk2020optimality} suggests that $\alpha \rightarrow 0$ as $N \rightarrow \infty$ meaning that FPWL can achieve sublinear regret for large system sizes. Simulations in both \cite{tripathi2019whittle} and \cite{maatouk2020optimality} indicate that $\alpha$ is very small for most problems of practical interest.

Importantly, note that there is no way to get sublinear static regret by using FPWL if the Whittle solution is not exactly optimal for the offline problem. In this case, even if the cost functions are the same in every epoch, there would be a small gap $\alpha > 0$ between the cost of the Whittle policy and the optimal policy in every epoch. The small constant gap will add up to give linear regret. Thus, the term $\alpha T$ in the regret bound above accounts for this cost of using an approximate optimization oracle rather than an exact one, and cannot be eliminated.

%However, sublinear regret can be achieved when comparing to a hypothetical Whittle-based algorithm that has access to cost functions in the current epoch. We call this algorithm Be-the-Whittle-Leader (BWL). Let cost functions in epoch $k$ be given by $f_k^{(1)},...,f_k^{(N)}$. Then, the scheduling policy chosen by BWL in epoch $k$ is given by:
%\begin{equation}
%	\pi_k^{\text{BWL}} = \text{Whittle} \bigg(\sum_{t=1}^{k} f_t^{(1)},...,\sum_{t=1}^{k} f_t^{(N)}\bigg).
%\end{equation}
%This algorithm requires knowledge of the AoI cost functions at epoch $k$ to compute $\pi_k^{\text{BWL}}$. The one epoch look-ahead means this is not a valid online learning algorithm. Further, a scheduler who knows the current cost functions can simply decide to choose the Whittle policy for the current epoch rather than using the cumulative sum of costs.

%The FPWL algorithm achieves sublinear static regret compared to the BWL algorithm. We present this as a corollary to Theorem \ref{thm:regret_multi}.
%\begin{framed}
%	\begin{corollary}
%		Consider any sequence of bounded AoI cost functions. Let the cost of using policy $\pi$ in epoch $k$ be denoted by $C_k(\pi)$. Set $\epsilon = \sqrt{\frac{2M}{N D^2 T}}$ and let the scheduling policy chosen by FPWL in epoch $k$ be $\pi_k^{\text{FPWL}}$. Then,
%		\begin{equation}
%			\mathbb{E}\bigg[  \sum_{k=1}^{T} \big\{C_k(\pi_k^{\text{FPWL}}) - C_k(\pi_k^{\text{BWL}}) \big\} \bigg] \leq 2D\sqrt{2MNT},
%		\end{equation}
%		where the expectation is taken over the random perturbations.
%	\end{corollary}
%\end{framed}

\subsubsection{Dynamic Regret}
A drawback of the online learning formulation is that sublinear regret is only possible when comparing to a simple class of policies since there are no constraints on the adversary choosing the cost functions. A more general notion of regret is \textit{dynamic regret} where cost is compared to an algorithm which chooses the best scheduling policy in \textit{each} epoch rather than the best fixed policy across epochs. Dynamic regret of an algorithm that chooses scheduling policy $\pi_k$ in epoch $k$ is defined as follows:
	\begin{equation}
	\label{eq:dyn_regret}
	\text{D-Regret}_T(\text{Alg.}, \mathcal{C}) \triangleq \sup_{C_{1,..,T} \in \mathcal{C}} \bigg\{  \sum_{k=1}^{T} C_k({\pi}_k) - \sum_{k=1}^{T} \min_{ \pi \in \Pi }  C_k(\pi) \bigg\},
	\end{equation}
where $\mathcal{C}$ incorporates constraints on the adversary. It is easy to show that if there are no constraints on how an adversary is allowed to choose the cost functions $C_1, ..., C_T$ then achieving sublinear dynamic regret is not possible. Thus, the definition of dynamic regret includes $\mathcal{C}$ which is the class of cost function sequences over which the regret is being considered and incorporates constraints on the adversary. %When cost functions do not change too often or too drastically across epochs, we expect to be able to design policies that achieve sublinear dynamic regret.

A number of recent works on online learning consider the problem of minimizing dynamic regret by constraining how the sequence of cost functions change over time (see \cite{besbes2015non,jadbabaie2015online,besbes2019optimal,cheung2019learning}). We follow the approach of \cite{besbes2015non} and \cite{besbes2019optimal} by defining the quantity ${V}_T$ which measures the variation of a given sequence of cost functions as follows:
\begin{equation}
	\sum_{k=2}^{T} \max_{\pi} \big|C_{k-1}(\pi) - C_k(\pi)\big| \leq V_T.
	\label{eq:mult_vt}
\end{equation}
%These works try to incorporate known information about the sequence of cost functions into variation budgets and then try to prove dynamic regret bounds that scale sublinearly with these budgets.
% \begin{equation}
% \label{eq:vt}
% 	\mathcal{V}_T(C_1,...,C_T) \triangleq \sum_{k=2}^{T} \max_{x \in \{1,...,M\}} \big|C_{k-1}(x) - C_k(x)\big|.
% \end{equation} 
% \begin{equation}
% 	{V}_T(C_1,...,C_T) \leq V_T.
% 	\label{eq:vt_single_def}
% \end{equation}
Suppose we know that any sequence of cost functions chosen by the adversary is going to satisfy the inequality \eqref{eq:mult_vt}. Then, we denote the set of allowable sequence of cost functions by $\mathcal{C}(V_T)$ and define the quantity $V_T$ as the variation budget given to the adversary.

%  Suppose that the cost functions are changing slowly with time. The total variation in costs is measured using $V_T$ similar to the single source setting. 

We can also use the Whittle index approach to achieve low dynamic regret. If $V_T$ is known to be sublinear in $T$ beforehand, then simply using the Whittle index policy for the cost functions revealed in the previous epoch is sufficient to get low dynamic regret. Specifically, set $f_0^{(i)} = \{1,...,M\}, \forall i \in \{1,...,N\}$ and let the scheduling policy in epochs $k$ be given by:
\begin{equation}
\label{eq:fdwl}
\pi_k = \text{Whittle} \bigg(f_{k-1}^{(1)},...,f_{k-1}^{(N)}\bigg).
\end{equation}
We call this algorithm \textit{Follow the Dynamic Whittle Leader} (FDWL).
\begin{lemma}
		%Consider the full feedback multiple source setting. 
		The dynamic regret of FDWL satisfies
		\begin{equation*}
		\text{D-Regret}_T(\text{FDWL},\mathcal{C}(V_T)) \leq \alpha T + V_T + D,
		\end{equation*}
		where $V_T$ is the variation budget as defined in \eqref{eq:mult_vt} and $D$ is the upper-bound on AoI cost functions.
		\label{lem:dyn_mult} 
\end{lemma}	
\begin{proof}
See Appendix \ref{pf:dyn_mult}.
\end{proof}

An important point to note here is that FDWL should only be used when an upper bound on $V_T$ that grows sublinearly with $T$ is known \textit{a priori}. If no such upper bound is known and $V_T$ grows linearly with $T$, then it can be shown that FDWL incurs static regret that is linear in $T$ meaning it performs worse than FPWL (Algorithm \ref{alg:FTWL}). This neatly splits the full feedback setting into two regimes. If $V_T$ is known to be sublinear use FDWL to get sublinear dynamic regret. Otherwise, use the entire history of cost functions as in FPWL to get sublinear static regret.

%\begin{IEEEproof}
%See Appendix \ref{pf:dyn_mult}.
%\end{IEEEproof}
Algorithm \ref{alg:FTWL} also highlights the strength of follow-the-leader style algorithms in solving online optimization problems with combinatorial structure. If a low complexity solution is known to the offline problem as with the Whittle index then it can be incorporated into FTPL as an optimization oracle. On the other hand, exponential weight update based algorithms like EXP3 \cite{auer2002nonstochastic} or EXP3.S \cite{besbes2019optimal} are standard in the bandit feedback case. Incorporating a Whittle index solution directly in these algorithms is not possible. This makes designing computationally efficient online learning algorithms for bandit feedback harder in the multiple source setting. We develop a heuristic solution for this below.

\subsubsection{Bandit Feedback}
%We now discuss partial and bandit feedback in this modified setting. As in the full feedback case, we assume that we receive information in the form of decouples cost functions of AoI. 
For bandit feedback, the cost function of AoI associated with source $i$ is only revealed during the time-slots in which it sends an update.%, consistent with our earlier notion of bandit feedback.
 Specifically, if at time-slot $j$ within epoch $k$ the policy $\pi_k$ schedules sensor $i$, then $f_k^{(i)}(A^{(i)}(j))$ is revealed to the scheduler. This happens for every time-slot in the epoch. 
%Thus, if the AoI process of the $i$th source during the $k$th epoch under scheduling policy $\pi_k$ was $1,2,...,h_1^{(i)},1,...,h_2^{(i)},...,1,...,h_l^{(i)},...$ where $l$ is the total number of transmissions by source $i$ in epoch $k$, then we assume that the scheduler is revealed values $f_k^{(i)}(h_1^{(i)}),...,f_k^{(i)}(h_l^{(i)})$ among the $M$ values of $f_k^{(i)}$.

To run FPWL and FDWL on this incomplete feedback we need to construct estimates of the cost functions denoted by $\hat{f}_k^{(1)},...,\hat{f}_k^{(N)}$. We do this by linearly interpolating between the revealed values of $f_k^{(i)}$ for each source $i$. Algorithm \ref{alg:c_estimate} describes the details. Importantly, constructing the linear interpolating cost estimates for a single source requires a single pass over $1,...,M$. Thus, constructing $\hat{f}_k^{(1)},...,\hat{f}_k^{(N)}$ has a computational complexity $O(NM)$. So, our modified versions of FPWL and FDWL for bandit feedback remain computationally efficient. However, since these estimates are not guaranteed to be unbiased, regret analysis in the bandit feedback case becomes challenging. %We leave this to future work.

\begin{algorithm}
	\DontPrintSemicolon
	%\SetAlgoLined
	%\KwResult{Write here the result}
	\SetKwInOut{Input}{Input}\SetKwInOut{Output}{Output}
	\Input{$X \subseteq \{1,...,M\}$ for which $f_k^{(i)}$ is known, $D$ known upper bound on $f_k^{(i)}$}
	\Output{Estimate $\hat{f}_k^{(i)}$ that is an increasing AoI cost function}
	\BlankLine
	Add $0$ to $X$ and set $f_k^{(i)}(0) = 0$\; 
	\If{$M \notin X$}{
		set $f_k^{(i)}(M) = D$ and add $M$ to $X$\;
	}
	Sort $X$ in increasing order $\{0,x_1,...,x_l,M\}$\; 		
	\While{ $h \in 1,...,M$ }{
		\eIf{$h \notin X$}{
		Find $k$ such that $x_k < h < x_{k+1}$ and $x_k,x_{k+1} \in X$\;
		Set $\hat{f}_k^{(i)}(h) = f_k^{(i)}(x_k) + (h-x_k)\frac{f_k^{(i)}(x_{k+1}) - f_k^{(i)}(x_k)}{x_{k+1} - x_k}$
		}{
		Set $\hat{f}_k^{(i)}(h) = f_k^{(i)}(h)$. } 				
	}		
	\caption{Linearly Interpolating Cost Function Estimate for source $i$}
	\label{alg:c_estimate}
\end{algorithm}

%Importantly, constructing the linear interpolating cost estimate for source $i$ requires a single pass over $1,...,M$ and hence has computational complexity $O(M)$. So estimates for all sources can be constructed efficiently in $O(NM)$ time. We use these estimates in lieu of the actual cost functions in Algorithm \ref{alg:FTWL} when full feedback is not available and the Algorithm still remains computationally efficient. We leave regret analysis of the FPWL algorithm with partial and bandit feedback for future work. This completes our discussion on online scheduling for multiple sources.
\section{Mobility Tracking}
\label{sec:applications}
%\begin{figure}
%	\centering
%	\includegraphics[width=0.85\linewidth]{Mobility.pdf}
%	%\caption{Mobility Tracking}
%	%\setlength{\belowcaptionskip}{-10pt}
%	\label{figure:mobility}
%\end{figure}
%It is important to note that learning to monitor or schedule is different from learning to estimate or control. Typical applications of our work are situations where the monitor already knows how to generate good estimates or control commands given possibly delayed information about the source, such as join the shortest queue or cache the $K$ most popular files. Our framework allows the monitor to learn how to schedule sources given a good estimation or control policy. 
We now apply the results we have developed to a mobility tracking problem. Consider $N$ nodes moving around in the two dimensional plane whose positions needs to be tracked by a central base station (BS). At any given time, only one of these nodes can send an update about its current location and velocity to the BS. The BS keeps track of the location of the nodes by storing the most recently received update from each node. Our goal is to design a scheduling policy that minimizes total tracking error between the location estimates at the BS and the actual locations of the nodes. %This is a common problem many real systems especially mobile networks and robotics applications.

Observe that if the current velocity of a node $i$ is $v_i$, then its tracking error grows linearly with its AoI. That is, if the BS hasn't received an update from node $i$ for time $A_i$, then the tracking error is $v_i A_i$. In practical scenarios, node mobility patterns and velocities are often unknown beforehand, non-stationary, and possibly adversarial. Thus, our mobility tracking problem can be viewed as a weighted-AoI minimization problem with time-varying velocities acting as weights. Since a new update from a node only contains information about its current location and velocity, \textit{this fits into the multiple source bandit feedback setting}.

We will discuss two specific mobility models and apply our online algorithms to show that they outperform static AoI based scheduling. Note that while cost functions being static within an epoch and resetting of AoIs at the beginning of every epoch are necessary for regret analysis, these assumptions are not required to implement our algorithms in practice. %In the following discussion we will consider real time-average costs and implement scheduling without resetting AoIs in every epoch.
\subsection{Levy Mobility}
 In this scenario, we simulate the nodes' motion using Levy mobility. This is a realistic mobility model that closely matches human mobility in practice \cite{rhee2011levy}. A node's motion is described in a series of \textit{steps}. A step is represented by the tuple $(v,\theta,t_f,t_p)$ - a velocity $v$ picked randomly in the interval $[0,v_{\text{max}}]$, an angle $\theta$ picked uniformly from the interval $[0,2\pi]$, a flight time $t_f$ picked at random from $\{1,...,{T_f}_{\text{max}}\}$ and a pause time $t_p$ picked at random from $\{1,...,{T_p}_{\text{max}}\}$. The node then moves with the velocity $v$, in the direction $\theta$ for $t_f$ time-slots and then pauses at its location for $t_p$ time-slots. This leads to a bursty random walk pattern with time-varying velocities.

\begin{figure}
	\centering
	\includegraphics[width=0.99\linewidth]{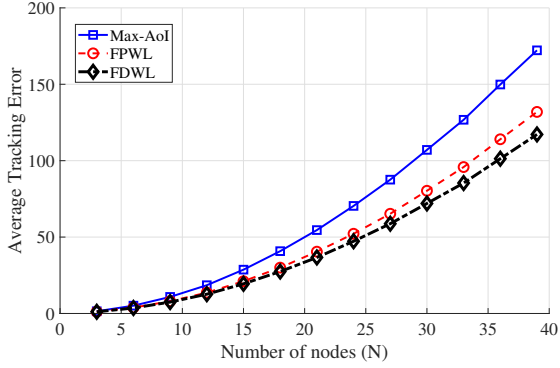}
	\caption{Levy Mobility: Average Tracking Error v/s number of nodes}
	\label{figure:levy}
\end{figure}
%${T_f}_{\text{max}} = 50$ and ${T_p}_{\text{max}}=30$ for all nodes,  epoch length $M = 200$, number of epochs $T = 500$
We consider $N$ nodes executing Levy mobility. An adversary sets the values of $v_{\text{max}}$ for each node from the set $\{0.1,0.5,5\}$ designating it as a slow, medium or fast node. Overall, $N/3$ nodes each are designated as fast, medium and slow, but the scheduler doesn't know which. We set ${T_f}_{\text{max}} = 50$ and ${T_p}_{\text{max}}=30$ for all nodes.

The scheduler does not know beforehand that there is inherent asymmetry in the motion of the nodes. An oblivious static scheduling policy is max-AoI: let the node with the maximum AoI transmit in every time-slot. From Figure \ref{figure:levy}, we observe that using FPWL in this setting outperforms the max-AoI scheduling policy (by about 25\%). Further, FDWL outperforms both max-AoI (by about 33\%) and FPWL (by about 10\%). This is because velocities under Levy mobility are slowly varying in time and not adversarial, allowing a dynamic regret based algorithm such as FDWL to work better than FPWL. We set the epoch length $M$ to 200 time-slots for both FPWL and FDWL, and the number of epochs $T$ to $500$, thus running the simulation for 100000 time-slots.

\subsection{Adversarial Mobility}
In this scenario, we assume that the nodes execute a mobility pattern that is chosen by a reactive adversary in response to the scheduling policies. In every epoch, the nodes execute Brownian motion (moving in random directions at a fixed velocity). An adversary assigns velocities to nodes such that they are inversely proportional to their scheduling priorities. 

For FPWL, the scheduling policy in epoch $t$ is given by $\pi_t = \text{Whittle} \big(F_t^{(1)} + \gamma_t^{(1)},...,F_t^{(N)} + \gamma_t^{(N)}\big)$. So, the velocity $v_t^{(i)}$ of node $i$ in epoch $t$ is chosen to satisfy 
$v_k^{(i)} \propto  c^{(i)}{||F_t^{(i)}||}^{-2}.$ Similarly, for FDWL, the scheduling policy in epoch $t$ is given by $\pi_t = \text{Whittle} \big(\hat{f}_{t-1}^{(1)},...,\hat{f}_{t-1}^{(N)}\big)$, where we use estimated cost functions since our setting involves bandit feedback. So, $v_t^{(i)}$ is chosen to satisfy
$v_k^{(i)} \propto  c^{(i)}{||\hat{f}_{t-1}^{(i)}||}^{-2}.$ For the max-AoI policy, the velocity $v_t^{(i)}$ is chosen to satisfy $v_k^{(i)} \propto  c^{(i)}.$ Here $c^{(i)}$ are parameters which are fixed across epochs and also chosen by the adversary to ensure that the motion of nodes has inherent asymmetry unknown to the scheduler. Overall, $N/3$ nodes each are assigned $c^{(i)} = 0.1$, $c^{(i)} = 0.4$ and $c^{(i)} = 40$. If the scheduler observes a node was moving fast in the previous epochs and assigns it a larger cost, then the adversary assigns it a low velocity in the next epoch so as to confuse the scheduler. The sum total of velocities is normalized and remains fixed in every time-slot ensuring that the \textit{adversary is equally powerful irrespective of scheduling policies}.
\begin{figure}
	\centering
	\includegraphics[width=0.93\linewidth]{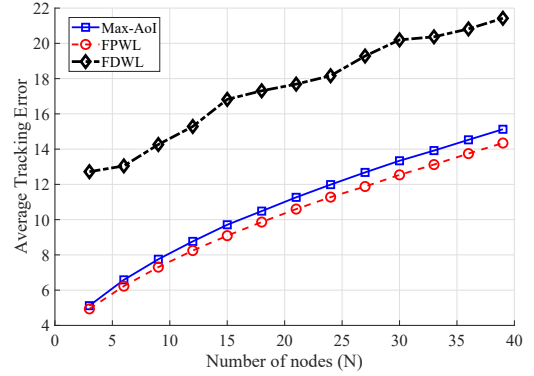}
	\caption{Adversarial Mobility: Average Tracking Error v/s number of nodes}
	\label{figure:adversarial}
\end{figure}

Under this adversarial model, we observe in Figure \ref{figure:adversarial} that while FPWL still outperforms max-AoI (by about 8\%), FDWL performs significantly worse than both FPWL and max-AoI (about 50\% worse). This is consistent with our results from theory - when cost functions are quickly varying and adversaries are unconstrained and reactive, dynamic regret based algorithms like FDWL perform worse than static regret algorithms like FPWL. %The epoch length $M = 200$ and the number of epochs $T = 500$, same as for Levy mobility.
%This also suggests an interesting extension - adaptively switching between FPWL and FDWL so as to get the best of both algorithms in a single online scheduling scheme. We leave this to future work.
\section{Conclusion}
\label{sec:conclusion}
In this work, we have formulated a general framework for online monitoring and scheduling for non-stationary sources. Specifically, we handle unknown, time-varying, and possibly adversarial cost functions of AoI and design algorithms that attempt to learn the best scheduling policies in an online fashion. We apply our results to a mobility tracking problem and show that our online learning algorithms outperform oblivious AoI based schemes and are able to learn information about the underlying source dynamics. 

Possible directions of future work involve applying our online scheduling framework to different problems of practical interest, and incorporating unreliable channels and noisy feedback about the costs into our framework.

% \begin{acks}
% \end{acks}

\bibliographystyle{ACM-Reference-Format}
\bibliography{bibliography}


\begin{thebibliography}{38}

%%% the Web address.

\ifx \showCODEN    \undefined \def \showCODEN     #1{\unskip}     \fi
\ifx \showDOI      \undefined \def \showDOI       #1{#1}\fi
\ifx \showISBNx    \undefined \def \showISBNx     #1{\unskip}     \fi
\ifx \showISBNxiii \undefined \def \showISBNxiii  #1{\unskip}     \fi
\ifx \showISSN     \undefined \def \showISSN      #1{\unskip}     \fi
\ifx \showLCCN     \undefined \def \showLCCN      #1{\unskip}     \fi
\ifx \shownote     \undefined \def \shownote      #1{#1}          \fi
\ifx \showarticletitle \undefined \def \showarticletitle #1{#1}   \fi
\ifx \showURL      \undefined \def \showURL       {\relax}        \fi
% invisible to TeX
\providecommand\bibfield[2]{#2}
\providecommand\bibinfo[2]{#2}
\providecommand\natexlab[1]{#1}
\providecommand\showeprint[2][]{arXiv:#2}

\bibitem[\protect\citeauthoryear{Auer, Cesa-Bianchi, Freund, and Schapire}{Auer
  et~al\mbox{.}}{2002}]%
        {auer2002nonstochastic}
\bibfield{author}{\bibinfo{person}{Peter Auer}, \bibinfo{person}{Nicolo
  Cesa-Bianchi}, \bibinfo{person}{Yoav Freund}, {and} \bibinfo{person}{Robert~E
  Schapire}.} \bibinfo{year}{2002}\natexlab{}.
\newblock \showarticletitle{The nonstochastic multiarmed bandit problem}.
\newblock \bibinfo{journal}{\emph{SIAM J. Comput.}} \bibinfo{volume}{32},
  \bibinfo{number}{1} (\bibinfo{year}{2002}), \bibinfo{pages}{48--77}.
\newblock


\bibitem[\protect\citeauthoryear{Banerjee, Bhattacharjee, and Sinha}{Banerjee
  et~al\mbox{.}}{2020}]%
        {banerjee2020fundamental}
\bibfield{author}{\bibinfo{person}{Subhankar Banerjee},
  \bibinfo{person}{Rajarshi Bhattacharjee}, {and} \bibinfo{person}{Abhishek
  Sinha}.} \bibinfo{year}{2020}\natexlab{}.
\newblock \showarticletitle{Fundamental limits of age-of-information in
  stationary and non-stationary environments}.
\newblock \bibinfo{journal}{\emph{arXiv preprint arXiv:2001.05471}}
  (\bibinfo{year}{2020}).
\newblock


\bibitem[\protect\citeauthoryear{Bedewy, Sun, and Shroff}{Bedewy
  et~al\mbox{.}}{2019}]%
        {bedewy2019minimizing}
\bibfield{author}{\bibinfo{person}{Ahmed~M Bedewy}, \bibinfo{person}{Yin Sun},
  {and} \bibinfo{person}{Ness~B Shroff}.} \bibinfo{year}{2019}\natexlab{}.
\newblock \showarticletitle{Minimizing the age of information through queues}.
\newblock \bibinfo{journal}{\emph{IEEE Trans. Information Theory}}
  \bibinfo{volume}{65}, \bibinfo{number}{8} (\bibinfo{year}{2019}),
  \bibinfo{pages}{5215--5232}.
\newblock


\bibitem[\protect\citeauthoryear{Besbes, Gur, and Zeevi}{Besbes
  et~al\mbox{.}}{2015}]%
        {besbes2015non}
\bibfield{author}{\bibinfo{person}{Omar Besbes}, \bibinfo{person}{Yonatan Gur},
  {and} \bibinfo{person}{Assaf Zeevi}.} \bibinfo{year}{2015}\natexlab{}.
\newblock \showarticletitle{Non-stationary stochastic optimization}.
\newblock \bibinfo{journal}{\emph{Operations research}} \bibinfo{volume}{63},
  \bibinfo{number}{5} (\bibinfo{year}{2015}), \bibinfo{pages}{1227--1244}.
\newblock


\bibitem[\protect\citeauthoryear{Besbes, Gur, and Zeevi}{Besbes
  et~al\mbox{.}}{2019}]%
        {besbes2019optimal}
\bibfield{author}{\bibinfo{person}{Omar Besbes}, \bibinfo{person}{Yonatan Gur},
  {and} \bibinfo{person}{Assaf Zeevi}.} \bibinfo{year}{2019}\natexlab{}.
\newblock \showarticletitle{Optimal exploration--exploitation in a multi-armed
  bandit problem with non-stationary rewards}.
\newblock \bibinfo{journal}{\emph{Stochastic Systems}} \bibinfo{volume}{9},
  \bibinfo{number}{4} (\bibinfo{year}{2019}), \bibinfo{pages}{319--337}.
\newblock


\bibitem[\protect\citeauthoryear{Bhandari, Fatale, Narula, Moharir, and
  Hanawal}{Bhandari et~al\mbox{.}}{2020}]%
        {bhandari2020age}
\bibfield{author}{\bibinfo{person}{Kavya Bhandari}, \bibinfo{person}{Santosh
  Fatale}, \bibinfo{person}{Urvidh Narula}, \bibinfo{person}{Sharayu Moharir},
  {and} \bibinfo{person}{Manjesh~Kumar Hanawal}.}
  \bibinfo{year}{2020}\natexlab{}.
\newblock \showarticletitle{Age-of-Information Bandits}.
\newblock \bibinfo{journal}{\emph{arXiv preprint arXiv:2001.09317}}
  (\bibinfo{year}{2020}).
\newblock


\bibitem[\protect\citeauthoryear{Cesa-Bianchi and Lugosi}{Cesa-Bianchi and
  Lugosi}{2006}]%
        {cesa2006prediction}
\bibfield{author}{\bibinfo{person}{Nicolo Cesa-Bianchi} {and}
  \bibinfo{person}{G{\'a}bor Lugosi}.} \bibinfo{year}{2006}\natexlab{}.
\newblock \bibinfo{booktitle}{\emph{Prediction, learning, and games}}.
\newblock \bibinfo{publisher}{Cambridge university press}.
\newblock


\bibitem[\protect\citeauthoryear{Champati, Mamduhi, Johansson, and
  Gross}{Champati et~al\mbox{.}}{2019}]%
        {champati2019performance}
\bibfield{author}{\bibinfo{person}{Jaya~Prakash Champati},
  \bibinfo{person}{Mohammad~H Mamduhi}, \bibinfo{person}{Karl~H Johansson},
  {and} \bibinfo{person}{James Gross}.} \bibinfo{year}{2019}\natexlab{}.
\newblock \showarticletitle{Performance characterization using aoi in a
  single-loop networked control system}. In \bibinfo{booktitle}{\emph{Proc.
  IEEE INFOCOM AoI Workshop}}. \bibinfo{pages}{197--203}.
\newblock


\bibitem[\protect\citeauthoryear{Cheung, Simchi-Levi, and Zhu}{Cheung
  et~al\mbox{.}}{2019}]%
        {cheung2019learning}
\bibfield{author}{\bibinfo{person}{Wang~Chi Cheung}, \bibinfo{person}{David
  Simchi-Levi}, {and} \bibinfo{person}{Ruihao Zhu}.}
  \bibinfo{year}{2019}\natexlab{}.
\newblock \showarticletitle{Learning to optimize under non-stationarity}. In
  \bibinfo{booktitle}{\emph{Proc. Int. Conf. Artificial Intell. Stats.
  (AISTATS)}}. \bibinfo{pages}{1079--1087}.
\newblock


\bibitem[\protect\citeauthoryear{Cohen and Hazan}{Cohen and Hazan}{2015}]%
        {cohen2015following}
\bibfield{author}{\bibinfo{person}{Alon Cohen} {and} \bibinfo{person}{Tamir
  Hazan}.} \bibinfo{year}{2015}\natexlab{}.
\newblock \showarticletitle{Following the perturbed leader for online
  structured learning}. In \bibinfo{booktitle}{\emph{Proc. Int. Conf. Machine
  Learning (ICML)}}. \bibinfo{pages}{1034--1042}.
\newblock


\bibitem[\protect\citeauthoryear{Farazi, Klein, McNeill, and Brown}{Farazi
  et~al\mbox{.}}{2018}]%
        {farazi2018age}
\bibfield{author}{\bibinfo{person}{Shahab Farazi}, \bibinfo{person}{Andrew~G
  Klein}, \bibinfo{person}{John~A McNeill}, {and} \bibinfo{person}{D~Richard
  Brown}.} \bibinfo{year}{2018}\natexlab{}.
\newblock \showarticletitle{On the age of information in multi-source multi-hop
  wireless status update networks}. In \bibinfo{booktitle}{\emph{Proc. IEEE
  Int. Workshop Signal Process. Adv. Wireless Commun. (SPAWC)}}.
  \bibinfo{pages}{1--5}.
\newblock


\bibitem[\protect\citeauthoryear{Hannan}{Hannan}{1957}]%
        {hannan1957approximation}
\bibfield{author}{\bibinfo{person}{James Hannan}.}
  \bibinfo{year}{1957}\natexlab{}.
\newblock \showarticletitle{Approximation to Bayes risk in repeated play}.
\newblock \bibinfo{journal}{\emph{Contributions to the Theory of Games}}
  \bibinfo{volume}{3} (\bibinfo{year}{1957}), \bibinfo{pages}{97--139}.
\newblock


\bibitem[\protect\citeauthoryear{Hazan}{Hazan}{2019}]%
        {hazan2019introduction}
\bibfield{author}{\bibinfo{person}{Elad Hazan}.}
  \bibinfo{year}{2019}\natexlab{}.
\newblock \showarticletitle{Introduction to online convex optimization}.
\newblock \bibinfo{journal}{\emph{arXiv preprint arXiv:1909.05207}}
  (\bibinfo{year}{2019}).
\newblock


\bibitem[\protect\citeauthoryear{Huang and Modiano}{Huang and Modiano}{2015}]%
        {huang2015optimizing}
\bibfield{author}{\bibinfo{person}{Longbo Huang} {and} \bibinfo{person}{Eytan
  Modiano}.} \bibinfo{year}{2015}\natexlab{}.
\newblock \showarticletitle{Optimizing age-of-information in a multi-class
  queueing system}. In \bibinfo{booktitle}{\emph{Proc. IEEE Int. Symp.
  Information Theory (ISIT)}}. \bibinfo{pages}{1681--1685}.
\newblock


\bibitem[\protect\citeauthoryear{Inoue, Masuyama, Takine, and Tanaka}{Inoue
  et~al\mbox{.}}{2018}]%
        {inoue2018general}
\bibfield{author}{\bibinfo{person}{Yoshiaki Inoue}, \bibinfo{person}{Hiroyuki
  Masuyama}, \bibinfo{person}{Tetsuya Takine}, {and} \bibinfo{person}{Toshiyuki
  Tanaka}.} \bibinfo{year}{2018}\natexlab{}.
\newblock \showarticletitle{A general formula for the stationary distribution
  of the age of information and its application to single-server queues}.
\newblock \bibinfo{journal}{\emph{arXiv preprint arXiv:1804.06139}}
  (\bibinfo{year}{2018}).
\newblock


\bibitem[\protect\citeauthoryear{Jadbabaie, Rakhlin, Shahrampour, and
  Sridharan}{Jadbabaie et~al\mbox{.}}{2015}]%
        {jadbabaie2015online}
\bibfield{author}{\bibinfo{person}{Ali Jadbabaie}, \bibinfo{person}{Alexander
  Rakhlin}, \bibinfo{person}{Shahin Shahrampour}, {and}
  \bibinfo{person}{Karthik Sridharan}.} \bibinfo{year}{2015}\natexlab{}.
\newblock \showarticletitle{Online optimization: Competing with dynamic
  comparators}. In \bibinfo{booktitle}{\emph{Proc. Int. Conf. Artificial
  Intell. Stats. (AISTATS)}}. \bibinfo{pages}{398--406}.
\newblock


\bibitem[\protect\citeauthoryear{Jhunjhunwala and Moharir}{Jhunjhunwala and
  Moharir}{2018}]%
        {jhun2018age}
\bibfield{author}{\bibinfo{person}{Prakirt~Raj Jhunjhunwala} {and}
  \bibinfo{person}{Sharayu Moharir}.} \bibinfo{year}{2018}\natexlab{}.
\newblock \showarticletitle{Age-of-Information Aware Scheduling}. In
  \bibinfo{booktitle}{\emph{Proc. IEEE SPCOM}}.
\newblock


\bibitem[\protect\citeauthoryear{Kadota, Sinha, and Modiano}{Kadota
  et~al\mbox{.}}{2019}]%
        {kadota2018scheduling2}
\bibfield{author}{\bibinfo{person}{Igor Kadota}, \bibinfo{person}{Abhishek
  Sinha}, {and} \bibinfo{person}{Eytan Modiano}.}
  \bibinfo{year}{2019}\natexlab{}.
\newblock \showarticletitle{Scheduling algorithms for optimizing age of
  information in wireless networks with throughput constraints}.
\newblock \bibinfo{journal}{\emph{IEEE/ACM Trans. Netw.}} \bibinfo{volume}{27},
  \bibinfo{number}{4} (\bibinfo{year}{2019}), \bibinfo{pages}{1359--1372}.
\newblock


\bibitem[\protect\citeauthoryear{Kadota, Sinha, Uysal-Biyikoglu, Singh, and
  Modiano}{Kadota et~al\mbox{.}}{2018}]%
        {kadota2018scheduling}
\bibfield{author}{\bibinfo{person}{Igor Kadota}, \bibinfo{person}{Abhishek
  Sinha}, \bibinfo{person}{Elif Uysal-Biyikoglu}, \bibinfo{person}{Rahul
  Singh}, {and} \bibinfo{person}{Eytan Modiano}.}
  \bibinfo{year}{2018}\natexlab{}.
\newblock \showarticletitle{Scheduling policies for minimizing age of
  information in broadcast wireless networks}.
\newblock \bibinfo{journal}{\emph{IEEE/ACM Trans. Netw.}} \bibinfo{volume}{26},
  \bibinfo{number}{6} (\bibinfo{year}{2018}), \bibinfo{pages}{2637--2650}.
\newblock


\bibitem[\protect\citeauthoryear{Kalai and Vempala}{Kalai and Vempala}{2005}]%
        {kalai2005efficient}
\bibfield{author}{\bibinfo{person}{Adam Kalai} {and} \bibinfo{person}{Santosh
  Vempala}.} \bibinfo{year}{2005}\natexlab{}.
\newblock \showarticletitle{Efficient algorithms for online decision problems}.
\newblock \bibinfo{journal}{\emph{J. Comput. System Sci.}}
  \bibinfo{volume}{71}, \bibinfo{number}{3} (\bibinfo{year}{2005}),
  \bibinfo{pages}{291--307}.
\newblock


\bibitem[\protect\citeauthoryear{Kam, Kompella, and Ephremides}{Kam
  et~al\mbox{.}}{2013}]%
        {kam2013age}
\bibfield{author}{\bibinfo{person}{Clement Kam}, \bibinfo{person}{Sastry
  Kompella}, {and} \bibinfo{person}{Anthony Ephremides}.}
  \bibinfo{year}{2013}\natexlab{}.
\newblock \showarticletitle{Age of information under random updates}. In
  \bibinfo{booktitle}{\emph{Proc. IEEE Int. Symp. Information Theory (ISIT)}}.
  \bibinfo{pages}{66--70}.
\newblock


\bibitem[\protect\citeauthoryear{Kam, Kompella, and Ephremides}{Kam
  et~al\mbox{.}}{2019}]%
        {kam2019learning}
\bibfield{author}{\bibinfo{person}{Clement Kam}, \bibinfo{person}{Sastry
  Kompella}, {and} \bibinfo{person}{Anthony Ephremides}.}
  \bibinfo{year}{2019}\natexlab{}.
\newblock \showarticletitle{Learning to sample a signal through an unknown
  system for minimum aoi}. In \bibinfo{booktitle}{\emph{Proc. IEEE INFOCOM AoI
  Workshop}}. \bibinfo{pages}{177--182}.
\newblock


\bibitem[\protect\citeauthoryear{Kaul, Yates, and Gruteser}{Kaul
  et~al\mbox{.}}{2012}]%
        {kaul2012real}
\bibfield{author}{\bibinfo{person}{Sanjit Kaul}, \bibinfo{person}{Roy Yates},
  {and} \bibinfo{person}{Marco Gruteser}.} \bibinfo{year}{2012}\natexlab{}.
\newblock \showarticletitle{Real-time status: How often should one update?}. In
  \bibinfo{booktitle}{\emph{Proc. IEEE INFOCOM}}. \bibinfo{pages}{2731--2735}.
\newblock


\bibitem[\protect\citeauthoryear{Kl{\"u}gel, Mamduhi, Hirche, and
  Kellerer}{Kl{\"u}gel et~al\mbox{.}}{2019}]%
        {klugel2019aoi}
\bibfield{author}{\bibinfo{person}{Markus Kl{\"u}gel},
  \bibinfo{person}{Mohammad~H Mamduhi}, \bibinfo{person}{Sandra Hirche}, {and}
  \bibinfo{person}{Wolfgang Kellerer}.} \bibinfo{year}{2019}\natexlab{}.
\newblock \showarticletitle{Aoi-penalty minimization for networked control
  systems with packet loss}. In \bibinfo{booktitle}{\emph{Proc. IEEE INFOCOM
  AoI Workshop}}. \bibinfo{pages}{189--196}.
\newblock


\bibitem[\protect\citeauthoryear{Kosta, Pappas, Angelakis, et~al\mbox{.}}{Kosta
  et~al\mbox{.}}{2017a}]%
        {kosta2017age}
\bibfield{author}{\bibinfo{person}{Antzela Kosta}, \bibinfo{person}{Nikolaos
  Pappas}, \bibinfo{person}{Vangelis Angelakis}, {et~al\mbox{.}}}
  \bibinfo{year}{2017}\natexlab{a}.
\newblock \showarticletitle{Age of information: A new concept, metric, and
  tool}.
\newblock \bibinfo{journal}{\emph{Foundations and Trends in Networking}}
  \bibinfo{volume}{12}, \bibinfo{number}{3} (\bibinfo{year}{2017}),
  \bibinfo{pages}{162--259}.
\newblock


\bibitem[\protect\citeauthoryear{Kosta, Pappas, Ephremides, and
  Angelakis}{Kosta et~al\mbox{.}}{2017b}]%
        {kosta2017nlage}
\bibfield{author}{\bibinfo{person}{Antzela Kosta}, \bibinfo{person}{Nikolaos
  Pappas}, \bibinfo{person}{Anthony Ephremides}, {and}
  \bibinfo{person}{Vangelis Angelakis}.} \bibinfo{year}{2017}\natexlab{b}.
\newblock \showarticletitle{Age and value of information: Non-linear age case}.
  In \bibinfo{booktitle}{\emph{Proc. IEEE Int. Symp. Information Theory
  (ISIT)}}. \bibinfo{pages}{326--330}.
\newblock


\bibitem[\protect\citeauthoryear{Maatouk, Kriouile, Assaad, and
  Ephremides}{Maatouk et~al\mbox{.}}{2020}]%
        {maatouk2020optimality}
\bibfield{author}{\bibinfo{person}{Ali Maatouk}, \bibinfo{person}{Saad
  Kriouile}, \bibinfo{person}{Mohamad Assaad}, {and} \bibinfo{person}{Anthony
  Ephremides}.} \bibinfo{year}{2020}\natexlab{}.
\newblock \showarticletitle{On The Optimality of The Whittle's Index Policy For
  Minimizing The Age of Information}.
\newblock \bibinfo{journal}{\emph{arXiv preprint arXiv:2001.03096}}
  (\bibinfo{year}{2020}).
\newblock


\bibitem[\protect\citeauthoryear{Ornee and Sun}{Ornee and Sun}{2019}]%
        {ornee2019sampling}
\bibfield{author}{\bibinfo{person}{Tasmeen~Zaman Ornee} {and}
  \bibinfo{person}{Yin Sun}.} \bibinfo{year}{2019}\natexlab{}.
\newblock \showarticletitle{Sampling for remote estimation through queues: Age
  of information and beyond}.
\newblock \bibinfo{journal}{\emph{IEEE Int. Symp. Model. Optim. Mobile, Ad Hoc
  Wireless Netw. (WiOpt)}} (\bibinfo{year}{2019}).
\newblock


\bibitem[\protect\citeauthoryear{Rhee, Shin, Hong, Lee, Kim, and Chong}{Rhee
  et~al\mbox{.}}{2011}]%
        {rhee2011levy}
\bibfield{author}{\bibinfo{person}{Injong Rhee}, \bibinfo{person}{Minsu Shin},
  \bibinfo{person}{Seongik Hong}, \bibinfo{person}{Kyunghan Lee},
  \bibinfo{person}{Seong~Joon Kim}, {and} \bibinfo{person}{Song Chong}.}
  \bibinfo{year}{2011}\natexlab{}.
\newblock \showarticletitle{On the levy-walk nature of human mobility}.
\newblock \bibinfo{journal}{\emph{IEEE/ACM Trans. Netw.}} \bibinfo{volume}{19},
  \bibinfo{number}{3} (\bibinfo{year}{2011}), \bibinfo{pages}{630--643}.
\newblock


\bibitem[\protect\citeauthoryear{Sun and Cyr}{Sun and Cyr}{2019}]%
        {sun2019sampling}
\bibfield{author}{\bibinfo{person}{Yin Sun} {and} \bibinfo{person}{Benjamin
  Cyr}.} \bibinfo{year}{2019}\natexlab{}.
\newblock \showarticletitle{Sampling for data freshness optimization:
  Non-linear age functions}.
\newblock \bibinfo{journal}{\emph{IEEE Journal Commun. Netw.}}
  \bibinfo{volume}{21}, \bibinfo{number}{3} (\bibinfo{year}{2019}),
  \bibinfo{pages}{204--219}.
\newblock


\bibitem[\protect\citeauthoryear{Sun, Kadota, Talak, and Modiano}{Sun
  et~al\mbox{.}}{2019}]%
        {sun2019age_book}
\bibfield{author}{\bibinfo{person}{Yin Sun}, \bibinfo{person}{Igor Kadota},
  \bibinfo{person}{Rajat Talak}, {and} \bibinfo{person}{Eytan Modiano}.}
  \bibinfo{year}{2019}\natexlab{}.
\newblock \showarticletitle{Age of information: A new metric for information
  freshness}.
\newblock \bibinfo{journal}{\emph{Synthesis Lectures on Communication
  Networks}} \bibinfo{volume}{12}, \bibinfo{number}{2} (\bibinfo{year}{2019}),
  \bibinfo{pages}{1--224}.
\newblock


\bibitem[\protect\citeauthoryear{Sun, Polyanskiy, and Uysal-Biyikoglu}{Sun
  et~al\mbox{.}}{2017a}]%
        {sun2017remote}
\bibfield{author}{\bibinfo{person}{Yin Sun}, \bibinfo{person}{Yury Polyanskiy},
  {and} \bibinfo{person}{Elif Uysal-Biyikoglu}.}
  \bibinfo{year}{2017}\natexlab{a}.
\newblock \showarticletitle{Remote estimation of the Wiener process over a
  channel with random delay}. In \bibinfo{booktitle}{\emph{Proc. IEEE Int.
  Symp. Information Theory (ISIT)}}. \bibinfo{pages}{321--325}.
\newblock


\bibitem[\protect\citeauthoryear{Sun, Uysal-Biyikoglu, Yates, Koksal, and
  Shroff}{Sun et~al\mbox{.}}{2017b}]%
        {yin17_tit_update_or_wait}
\bibfield{author}{\bibinfo{person}{Y. Sun}, \bibinfo{person}{E.
  Uysal-Biyikoglu}, \bibinfo{person}{R.~D. Yates}, \bibinfo{person}{C.~E.
  Koksal}, {and} \bibinfo{person}{N.~B. Shroff}.}
  \bibinfo{year}{2017}\natexlab{b}.
\newblock \showarticletitle{Update or Wait: How to Keep Your Data Fresh}.
\newblock \bibinfo{journal}{\emph{IEEE Trans. Information Theory}}
  \bibinfo{volume}{63}, \bibinfo{number}{11} (\bibinfo{date}{Nov.}
  \bibinfo{year}{2017}), \bibinfo{pages}{7492--7508}.
\newblock


\bibitem[\protect\citeauthoryear{Talak, Karaman, and Modiano}{Talak
  et~al\mbox{.}}{2018}]%
        {talak2018optimizing}
\bibfield{author}{\bibinfo{person}{Rajat Talak}, \bibinfo{person}{Sertac
  Karaman}, {and} \bibinfo{person}{Eytan Modiano}.}
  \bibinfo{year}{2018}\natexlab{}.
\newblock \showarticletitle{Optimizing information freshness in wireless
  networks under general interference constraints}. In
  \bibinfo{booktitle}{\emph{Proc. ACM Int. Symp. Mobile Ad Hoc Netw. Comput.
  (MobiHoc)}}. \bibinfo{pages}{61--70}.
\newblock


\bibitem[\protect\citeauthoryear{Tripathi and Modiano}{Tripathi and
  Modiano}{2019}]%
        {tripathi2019whittle}
\bibfield{author}{\bibinfo{person}{Vishrant Tripathi} {and}
  \bibinfo{person}{Eytan Modiano}.} \bibinfo{year}{2019}\natexlab{}.
\newblock \showarticletitle{A whittle index approach to minimizing functions of
  age of information}. In \bibinfo{booktitle}{\emph{Proc. 57th Allerton Conf.
  Commun. Control Comput.}} IEEE, \bibinfo{pages}{1160--1167}.
\newblock


\bibitem[\protect\citeauthoryear{Tripathi and Moharir}{Tripathi and
  Moharir}{2017}]%
        {tripathi2017age}
\bibfield{author}{\bibinfo{person}{Vishrant Tripathi} {and}
  \bibinfo{person}{Sharayu Moharir}.} \bibinfo{year}{2017}\natexlab{}.
\newblock \showarticletitle{Age of information in multi-source systems}. In
  \bibinfo{booktitle}{\emph{Proc. IEEE Global Commun. Conf. (GLOBECOM)}}.
  \bibinfo{pages}{1--6}.
\newblock


\bibitem[\protect\citeauthoryear{Yun, Joo, and Eryilmaz}{Yun
  et~al\mbox{.}}{2018}]%
        {yun2018optimal}
\bibfield{author}{\bibinfo{person}{Jihyeon Yun}, \bibinfo{person}{Changhee
  Joo}, {and} \bibinfo{person}{Atilla Eryilmaz}.}
  \bibinfo{year}{2018}\natexlab{}.
\newblock \showarticletitle{Optimal real-time monitoring of an information
  source under communication costs}. In \bibinfo{booktitle}{\emph{IEEE Conf.
  Decis. Control (CDC)}}. \bibinfo{pages}{4767--4772}.
\newblock


\bibitem[\protect\citeauthoryear{Zheng, Zhou, and Niu}{Zheng
  et~al\mbox{.}}{2019}]%
        {zheng2019context}
\bibfield{author}{\bibinfo{person}{Xi Zheng}, \bibinfo{person}{Sheng Zhou},
  {and} \bibinfo{person}{Zhisheng Niu}.} \bibinfo{year}{2019}\natexlab{}.
\newblock \showarticletitle{Context-aware information lapse for timely status
  updates in remote control systems}. In \bibinfo{booktitle}{\emph{Proc. IEEE
  Global Commun. Conf. (GLOBECOM)}}. \bibinfo{pages}{1--6}.
\newblock


\end{thebibliography}

\appendix
%\section{}
\section{Proof of Lemma \ref{lem:single_epoch}}
\label{pf:single_epoch}
Let the AoI cost function in epoch $k$ be $f_k(\cdot)$, let the transmission cost be $C$ and let the chosen sampling threshold be $x$. We set $t=1$ at the beginning of the epoch. Then,
\begin{equation}
    C_k(x) = \sum_{t=1}^{M} f_k(A(t)) + C u(t).
\end{equation}

Note that the AoI at time $t=1$ is $A(1) = 1$, since each epoch begins after a new transmission. Since the threshold is set to $x$, no new update is sent till time-slot $x$ at which point the AoI reaches $x$. Now, a new sample is generated and sent, so the AoI drops to 1 in the next time-slot. This process repeats in cycles of $x$ time-slots. Since the epoch consists of $M$ time-slots, there are $\lfloor \frac{M}{x} \rfloor$ complete cycles of length $x$. The sum of costs over each of these cycles is $\big(\sum_{j=1}^{x}f_k(j) + C\big)$ since the AoI goes from 1 to $x$ and there is a transmission at the end.

The final cycle is of length $r = M\mod x$ where $a\mod b$ is the remainder when $a$ is divided by $b$. There is a mandatory transmission in the final time-slot regardless of the AoI exceeding the threshold to finish the epoch. Thus,  
\begin{equation}
    C_k(x) = \bigg\lfloor \frac{M}{x} \bigg\rfloor \bigg(\sum_{j=1}^{x}f_k(j) + C\bigg) + \vmathbb{1}_{r > 0} \bigg(\sum_{j=1}^{r}f_k(j) + C\bigg).
\end{equation}
This completes the proof.

\section{Proof of Corollary \ref{corr:1}}
\label{pf:corr}
Let the regret of algorithm $\pi$ be $f(M,T)$. From Lemma \ref{lem:epoch_time}, we know that for any bounded sequence $f_1,...,f_T$ 
%\begin{equation*}
    \begin{multline*}
        \mathbb{E}\Bigg[ \bigg\{  \sum_{k=1}^{T} \sum_{t \in E_k} f_k(A^{\pi}(t)) + C u^{\pi}(t) -\\ \sum_{k=1}^{T} \sum_{t \in E_k} f_k(A^{\pi^{*}}(t)) + C u^{\pi^{*}}(t) \bigg\}\Bigg] \leq f(M,T).
    \end{multline*}
Dividing the equation about by $MT$, we get
\begin{multline*}
    \frac{1}{MT} \mathbb{E}\bigg[ \sum_{k=1}^{T}\sum_{t \in E_k} f_k(A^{\pi}(t)) + C u^{\pi}(t)  \bigg] \leq \\ \frac{1}{MT} \mathbb{E}\bigg[ \sum_{k=1}^{T}\sum_{t \in E_k} f_k(A^{\pi^{*}}(t)) + C u^{\pi^{*}}(t)  \bigg] + \frac{f(M,T)}{MT}.
\end{multline*}
Taking the limit supremum as $T$ goes to infinity and using the fact that $f(M,T)$ grows sublinearly in $T$, we get the required result.

\section{Closeness of Whittle and Optimal Policies}
\label{pf:alpha}
Here, we define $\alpha$, the parameter that measures the closeness of the Whittle index policy to an optimal policy within an epoch.

Consider a set of monotone and bounded AoI cost functions $f^{(1)},...,f^{(N)}$ such that for all $i$, if $x > y$ then $f^{(i)}(x) \geq f^{(i)}(y)$ and $f^{(i)}(M) \leq D$. Let $\text{Whittle}(f)$ denote the Whittle policy for this set of cost functions, as defined in \eqref{eq:whittle_policy}. Let $\text{Opt}(f)$ denote an optimal policy for this set of cost functions.

Now consider another set of monotone bounded AoI cost functions $g^{(1)},...,g^{(N)}$ with the same upper bound $D$. Given a scheduling policy $\pi$, let 
\begin{equation}
C_g(\pi) \triangleq
\frac{1}{NM}\sum\limits_{j=1}^{M} \sum\limits_{i=1}^{N} g^{(i)}(A^{(i)}(j)),
\end{equation}
where the AoIs evolve under policy $\pi$. This is the total sum cost of policy $\pi$ under the cost functions $g^{(1)},...,g^{(N)}$. We make the following assumption on the structure of Whittle index and optimal policies when the epoch length $M$ is long.
%\begin{framed}
	\begin{assumption}
		For any two sets of bounded monotone sets of cost functions $f^{(1)},...,f^{(N)}$ and $g^{(1)},...,g^{(N)}$ with a fixed known upper bound $D$, the following holds:
		\begin{equation}
		\bigg|C_g\big(\text{Whittle}(f)\big) - C_g\big(\text{Opt}(f)\big)\bigg| \leq \alpha,
		\end{equation}
		where $\alpha$ is a small constant that can depend on $N$, $M$ and $D$.	
		\label{ass:whittle_close}	
	\end{assumption}
%\end{framed} 

Note that this assumption is stronger than just assuming that the Whittle index policy has near optimal performance over long epochs. We assume that the Whittle policy is also close to the optimal policy in its sequence of scheduling decisions. Thus, given arbitrary bounded cost functions, the two policies $C_g\big(\text{Whittle}(f)\big)$ and $C_g\big(\text{Opt}(f)\big)$  have average costs that are close to each other. This is a Lipschitz like assumption on the policy space and cost functions for the scheduling problem. The motivation for this comes from results in \cite{tripathi2019whittle}, where it was shown that the Whittle policy is exactly optimal for $N=2$ as $M \rightarrow \infty$, meaning that we can set $\alpha = 0$. It was also observed via simulations that the Whittle policies are structurally similar to optimal policies for general $N$. Results on asymptotic optimality of the Whittle policy \cite{maatouk2020optimality} further suggest that $\alpha \rightarrow 0$ as $N \rightarrow \infty$.

\section{Proof of Theorem \ref{thm:regret_multi}}
\label{pf:w_regret}
Suppose $f_k^{(1)},...,f_k^{(N)}$ are the AoI cost functions during epoch $k$. In each epoch, the cost functions $f_k^{(i)}:\{1,...,M\}\rightarrow\mathbb{R}^{+}$ are  bounded monotone increasing functions of AoI, i.e. if $x > y$ then $f_k^{(i)}(x) \geq f_k^{(i)}(y)$ and $f_k^{(i)}(\cdot) \leq D$. $D$ is fixed and known beforehand. Let $C_k(\pi)$ be the cost incurred in epoch $k$ by using scheduling policy $\pi$, given by \eqref{eq:split_cost}. For a set of cost functions $f^{(1)},...,f^{(N)}$, the Whittle scheduling policy is represented by $\text{Whittle}(f^{(1)},...,f^{(N)})$. For the same set of cost functions, an optimal scheduling policy is represented by $\text{Opt}(f^{(1)},...,f^{(N)})$. We will use these notations throughout the proof. 

Similar to \cite{kalai2005efficient}, we will divide our proof into three steps.
\subsection{Be-the-Whittle-Leader has low regret}

First, we define a hypothetical algorithm called Be-the-Whittle-Leader (BWL). In epoch $k$, a scheduling policy $\pi_k^{\text{BWL}}$ is chosen as follows:
\begin{equation}
	\pi_k^{\text{BWL}} = \text{Whittle} \bigg(\sum_{t=1}^{k} f_t^{(1)},...,\sum_{t=1}^{k} f_t^{(N)}\bigg).
\end{equation}
BWL applies the Whittle procedure to the sum of cost functions seen from epoch $1$ through $k$ and uses this as the scheduling policy in epoch $k$. Clearly, this requires knowledge of the cost functions in the current epoch $k$ and hence, it is not an online learning algorithm. In this step, we will show that this algorithm, which looks ahead one epoch into the future, achieves low regret. In the next two steps, we will show that the gap between FPWL and BWL increases only sublinearly in $T$, completing the proof.

Note from \eqref{eq:whittle_policy} that if all cost functions are multiplied by a fixed positive constant, the Whittle and optimal policies remain unchanged. So, we rewrite BWL as:
\begin{equation}
\pi_k^{\text{BWL}} = \text{Whittle} \bigg(\frac{1}{k}\sum_{t=1}^{k} f_t^{(1)},...,\frac{1}{k}\sum_{t=1}^{k} f_t^{(N)}\bigg).
\end{equation}
Since AoI cost functions in each epoch are upper-bounded by $D$, their averages are also upper-bounded by $D$. Thus, we can apply Assumption \ref{ass:whittle_close} to the BWL policy. This results in the following inequality $\forall k \in 1,...,T$
\begin{equation}
C_k(\pi_k^{\text{BWL}}) \leq C_k\Bigg(\text{Opt}\bigg(\frac{1}{k}\sum_{t=1}^{k} f_t^{(1)},...,\frac{1}{k}\sum_{t=1}^{k} f_t^{(N)}\bigg)\Bigg) + \alpha.
\end{equation}
Summing the equation above for $k=1,...,T$, we get
\begin{equation}
\sum_{k=1}^{T} C_k(\pi_k^{\text{BWL}}) \leq \sum_{k=1}^{T} C_k\Bigg(\text{Opt}\bigg(\frac{1}{k}\sum_{t=1}^{k} f_t^{(1)},...,\frac{1}{k}\sum_{t=1}^{k} f_t^{(N)}\bigg)\Bigg) + \alpha T.
\label{eq:fpwl_5}
\end{equation}
Now, we claim that 
\begin{equation}
	\sum_{k=1}^{T} C_k\Bigg(\text{Opt}\bigg(\frac{1}{k}\sum_{t=1}^{k} f_t^{(1)},...,\frac{1}{k}\sum_{t=1}^{k} f_t^{(N)}\bigg)\Bigg) \leq \min_{\pi \in \Pi}\sum_{k=1}^{T} C_k\big(\pi\big).
	\label{eq:fpwl_4}
\end{equation}
To prove this, we use induction. For the base case, observe that the following holds by the definition of $\text{Opt}(\cdot)$.
\begin{equation}
	C_1\bigg(\text{Opt}(f_1^{(1)},...,f_N^{(1)})\bigg) = \min_{\pi \in \Pi} C_1(\pi)
\end{equation}
Further, since costs across epochs are additive, we have the following for any $l \in 1,...,T$:
\begin{equation}
	\sum_{k=1}^{l} C_k\Bigg(\text{Opt}\bigg(\frac{1}{l}\sum_{t=1}^{l} f_t^{(1)},...,\frac{1}{l}\sum_{t=1}^{l} f_t^{(N)}\bigg)\Bigg) = \min_{\pi \in \Pi} \sum_{k=1}^{l} C_k\big(\pi\big).
	\label{eq:fpwl_1}
\end{equation}
The above equation simply states that a policy that is optimal for the sum of cost functions from $1,...,l$ is also the best fixed scheduling policy to be used over the epochs $1,...,l$. 

Let's assume the following for some $l$:
\begin{equation}
		\sum_{k=1}^{l} C_k\Bigg(\text{Opt}\bigg(\frac{1}{k}\sum_{t=1}^{k} f_t^{(1)},...,\frac{1}{k}\sum_{t=1}^{k} f_t^{(N)}\bigg)\Bigg) \leq \min_{\pi \in \Pi}\sum_{k=1}^{l} C_k\big(\pi\big).
		\label{eq:fpwl_3}
\end{equation}
Then, adding the term $C_{l+1}\Bigg(\text{Opt}\bigg(\frac{1}{l+1}\sum\limits_{t=1}^{l+1} f_t^{(1)},...,\frac{1}{l+1}\sum\limits_{t=1}^{l+1} f_t^{(N)}\bigg)\Bigg)$ to both sides we get:
\begin{multline}
	\sum_{k=1}^{l+1} C_k\Bigg(\text{Opt}\bigg(\frac{1}{k}\sum_{t=1}^{k} f_t^{(1)},...,\frac{1}{k}\sum_{t=1}^{k} f_t^{(N)}\bigg)\Bigg) \leq \min_{\pi \in \Pi} \bigg \{\sum_{k=1}^{l} C_k\big(\pi\big) \bigg\} + \\ C_{l+1}\Bigg(\text{Opt}\bigg(\frac{1}{l+1}\sum_{t=1}^{l+1} f_t^{(1)},...,\frac{1}{l+1}\sum_{t=1}^{l+1} f_t^{(N)}\bigg)\Bigg).
\end{multline}
Note that the first term in the RHS is a minimum over all policies, so it can be upper bounded by replacing $\pi$ with any policy. This implies: 
\begin{multline}
\sum_{k=1}^{l+1} C_k\Bigg(\text{Opt}\bigg(\frac{1}{k}\sum_{t=1}^{k} f_t^{(1)},...,\frac{1}{k}\sum_{t=1}^{k} f_t^{(N)}\bigg)\Bigg) \leq \\ \sum_{k=1}^{l+1} C_{k}\Bigg(\text{Opt}\bigg(\frac{1}{l+1}\sum_{t=1}^{l+1} f_t^{(1)},...,\frac{1}{l+1}\sum_{t=1}^{l+1} f_t^{(N)}\bigg)\Bigg).
\end{multline}
Using \eqref{eq:fpwl_1} we can rewrite this as:
\begin{equation}
	\sum_{k=1}^{l+1} C_k\Bigg(\text{Opt}\bigg(\frac{1}{k}\sum_{t=1}^{k} f_t^{(1)},...,\frac{1}{k}\sum_{t=1}^{k} f_t^{(N)}\bigg)\Bigg) \leq \min_{\pi \in \Pi}\sum_{k=1}^{l+1} C_k\big(\pi\big).
	\label{eq:fpwl_2}
\end{equation}
Thus, assuming \eqref{eq:fpwl_3}, we were able to prove \eqref{eq:fpwl_2}. By induction on $l$, this proves \eqref{eq:fpwl_4}. Combining \eqref{eq:fpwl_4} with \eqref{eq:fpwl_5}, we get:
\begin{equation}
\sum_{k=1}^{T} C_k(\pi_k^{\text{BWL}}) \leq \min_{\pi \in \Pi}\sum_{k=1}^{T} C_k\big(\pi\big) + \alpha T.
\label{eq:fpwl_6}
\end{equation}
Finally, \eqref{eq:fpwl_6} together with the definition of static regret \eqref{eq:regret} implies that:
\begin{equation}
	\text{Regret}_T(\text{BWL}) \leq \alpha T.
	\label{eq:bwl}
\end{equation}

\subsection{Be-the-Perturbed-Whittle-Leader has low regret} Now, we consider a policy called Be-the-Perturbed-Whittle-Leader (BPWL). This is similar to the BWL policy, but it involves adding an extra perturbation to the cost functions before computing the Whittle index.

We first describe how the perturbation is generated. First, we generate $NM$ i.i.d. random variables $\delta^{(i)}(j) \sim \text{Uniform}\big([0,1/\epsilon]\big), \forall i \in 1,...,N \text{ and }\forall j \in 1,...,M$. We collect these random variables into $N$ vectors $\delta^{(1)},...,\delta^{(N)}$, where each vector $\delta^{(i)} \in \mathbb{R}^{M}$. Using these, we create monotonically increasing random vectors $\gamma^{(1)},...,\gamma^{(N)}$ as follows:
\begin{equation}
	\gamma^{(i)}(j) = \sum_{k=1}^{j} \delta^{(i)}(k), \forall i \in 1,...,N \text{ and }\forall j \in 1,...,M.
\end{equation}
Now, we have $N$ $M$-dimensional random vectors that are monotonically increasing. Given any set of AoI cost functions $f^{(1)},...,f^{(N)}$, the perturbation procedure is given by:
\begin{equation}
	\text{Perturb}\bigg(f^{(1)},...,f^{(N)}\bigg) = \bigg(f^{(1)} + \gamma^{(1)},...,f^{(N)}  + \gamma^{(N)} \bigg).
\end{equation}

Now, we can describe the hypothetical algorithm called Be-the-Perturbed-Whittle-Leader (BPWL). In epoch $k$, a scheduling policy $\pi_k^{\text{BPWL}}$ is chosen as follows:
\begin{equation}
\pi_k^{\text{BPWL}} = \text{Whittle}\bigg( \text{Perturb}\bigg(\sum_{t=1}^{k} f_t^{(1)},...,\sum_{t=1}^{k} f_t^{(N)}\bigg) \bigg),
\label{eq:bpwl_def}
\end{equation}
where the perturbations are generated i.i.d. for every epoch $k$. We denote the the perturbations in epoch $k$ by $\gamma^{(1)}_k,...,\gamma^{(N)}_k$. Since $\gamma^{(1)}_k,...,\gamma^{(N)}_k$ are monotone increasing functions, they can themselves be viewed as AoI costs. The cost of a policy $\pi$ with the AoI cost functions $\gamma^{(1)}_k,...,\gamma^{(N)}_k$ is denoted by $C_{\gamma_k}(\pi)$. We will use this notation later.

Now, consider a sequence such that in epoch $k$, the AoI cost functions are given by:
\begin{equation}
	\bigg( \tilde{f}_k^{(1)},...,\tilde{f}_k^{(N)} \bigg) = \bigg(f_k^{(1)} + \gamma_k^{(1)} - \gamma_{k-1}^{(1)},...,f_k^{(N)}  + \gamma_k^{(N)} - \gamma_{k-1}^{(N)} \bigg),
	\label{eq:bpwl_4}
\end{equation}
where $\gamma_{0}^{(i)} = \mathbf{0}$ for all $i$. Let $\tilde{C}_k(\pi)$ denote the cost of using scheduling policy $\pi$ in epoch $k$ where the AoI cost functions are $\tilde{f}_k^{(1)},...,\tilde{f}_k^{(N)}$. 

Observe that the cumulative cost functions in epoch $k$ for this hypothetical sequence are given by:
\begin{equation}
\begin{split}
    \bigg( \sum_{t=1}^{k}\tilde{f}_t^{(1)},...,\sum_{t=1}^{k} \tilde{f}_t^{(N)} \bigg) &= \bigg(\sum_{t=1}^{k} f_t^{(1)} + \gamma_k^{(1)},...,\sum_{t=1}^{k} f_t^{(N)}  + \gamma_k^{(N)} \bigg) \\ &= \text{Perturb}\bigg(\sum_{t=1}^{k} f_t^{(1)},...,\sum_{t=1}^{k} f_t^{(N)}\bigg).
\end{split}
\label{eq:bpwl_2}
\end{equation}
Because of the way the perturbations are created the cumulative cost functions $\bigg( \sum_{t=1}^{k}\tilde{f}_t^{(1)},...,\sum_{t=1}^{k} \tilde{f}_t^{(N)} \bigg)$ are monotone increasing functions of AoI in every epoch $k$. Thus, we can apply \eqref{eq:fpwl_6} to this sequence of cost functions to get:
\begin{equation}
\sum_{k=1}^{T} \tilde{C}_k\bigg(\text{Whittle} \bigg(\sum_{t=1}^{k} \tilde{f}_t^{(1)},...,\sum_{t=1}^{k} \tilde{f}_t^{(N)}\bigg)\bigg) \leq \min_{\pi \in \Pi}\sum_{k=1}^{T} \tilde{C}_k\big(\pi\big) + \alpha T.
\label{eq:bpwl_1}
\end{equation}
Now using \eqref{eq:bpwl_2} and the definition of BPWL \eqref{eq:bpwl_def}, we get:
\begin{equation}
\sum_{k=1}^{T} \tilde{C}_k(\pi_k^{\text{BPWL}}) \leq \min_{\pi \in \Pi}\sum_{k=1}^{T} \tilde{C}_k\big(\pi\big) + \alpha T.
\label{eq:bpwl_3}
\end{equation}
Observe that the first term in the RHS is a minimization over all policies $\pi$, so we can replace $\pi$ with $\text{Opt}\bigg(\sum_{t=1}^{T} f_t^{(1)},...,\sum_{t=1}^{T} f_t^{(N)}\bigg)$. This is the best fixed scheduling policy for the \textit{original} sequence of cost functions.
\begin{equation}
	\sum_{k=1}^{T} \tilde{C}_k(\pi_k^{\text{BPWL}}) \leq \sum_{k=1}^{T} \tilde{C}_k \bigg( \text{Opt}\bigg(\sum_{t=1}^{T} f_t^{(1)},...,\sum_{t=1}^{T} f_t^{(N)}\bigg) \bigg) + \alpha T.
	\label{eq:bpwl_5}
\end{equation}
Note that costs across epochs are additive. So, using \eqref{eq:bpwl_4} for any fixed policy $\pi$, we get:
\begin{equation}
	\sum_{k=1}^{T} \tilde{C}_k (\pi) = \sum_{k=1}^{T} \bigg(C_k (\pi) + C_{\gamma_k}(\pi) - C_{\gamma_{k-1}}(\pi) \bigg).
	\label{eq:bpwl_7}
\end{equation}
This further simplifies to:
\begin{equation}
\sum_{k=1}^{T} \tilde{C}_k (\pi) = \sum_{k=1}^{T} \bigg(C_k (\pi)\bigg) + C_{\gamma_T}(\pi).
\label{eq:bpwl_6}
\end{equation}
Applying \eqref{eq:bpwl_6} to \eqref{eq:bpwl_5} and using the definition of $\text{Opt}(\cdot)$ we get:
\begin{equation}
	\sum_{k=1}^{T} \tilde{C}_k(\pi_k^{\text{BPWL}}) \leq 
	\min_{\pi \in \Pi}\sum_{k=1}^{T} C_k\big(\pi\big) + \max_{\pi \in \Pi}C_{\gamma_T}(\pi) + \alpha T. 
	\label{eq:bpwl_8}
\end{equation}
Using \eqref{eq:bpwl_4}, we can also conclude that:
\begin{multline}
	\sum_{k=1}^{T} C_k(\pi_k^{\text{BPWL}}) \leq \sum_{k=1}^{T} \tilde{C}_k(\pi_k^{\text{BPWL}})  + \\ \sum_{k=1}^{T} \bigg|  C_{\gamma_k}(\pi_k^{\text{BPWL}}) - C_{\gamma_{k-1}}(\pi_k^{\text{BPWL}})   \bigg|
	\label{eq:bpwl_9}
\end{multline}
Combining \eqref{eq:bpwl_8} and \eqref{eq:bpwl_9}, we get:
\begin{multline}
	\sum_{k=1}^{T} C_k(\pi_k^{\text{BPWL}})  \leq \min_{\pi \in \Pi}\sum_{k=1}^{T} C_k\big(\pi\big) + \\ \sum_{k=1}^{T} \max_{\pi \in \Pi} \bigg|  C_{\gamma_k}(\pi) - C_{\gamma_{k-1}}(\pi)   \bigg| + \max_{\pi \in \Pi} C_{\gamma_T}(\pi) + \alpha T. 
	\label{eq:bpwl_10}
\end{multline}
Now, we will use a trick that is standard in online learning literature. We will assume that the adversary choosing the sequence of bounded cost functions is non-reactive, i.e the sequence of cost functions is chosen in advance. Thus, for the purposes of expected regret, it is sufficient to use the same perturbations $\gamma^{(1)}_1,...,\gamma^{(N)}_1$ in every epoch (since the adversary cannot learn the perturbations). For this choice of perturbations, \eqref{eq:bpwl_10} simplifies to:
\begin{equation}
	\mathbb{E}\bigg[\sum_{k=1}^{T} C_k(\pi_k^{\text{BPWL}})\bigg]  \leq \min_{\pi \in \Pi}\sum_{k=1}^{T} C_k\big(\pi\big) + 2\max_{\pi \in \Pi} C_{\gamma_1}(\pi) + \alpha T.
\end{equation}  
Observe that the maximum value that $\gamma_1^{(i)}(j)$ can have for any value of $i$ and $j$ is $M/\epsilon$. Thus, by the definition of average cost in an epoch \eqref{eq:split_cost}, we know that:
\begin{equation}
	\max_{\pi \in \Pi} C_{\gamma_1}(\pi) \leq \frac{M}{\epsilon}.
\end{equation}
Putting everything together, we have:
\begin{equation}
\mathbb{E}\big[\text{Regret}_T(\text{BPWL})\big] \leq 2\frac{M}{\epsilon} + \alpha T.
\label{eq:bpwl}
\end{equation}
While we proved this by assuming an oblivious adversary, the extension to a reactive or non-oblivious adversary is straightforward from Lemma 4.1 in \cite{cesa2006prediction}.

\subsection{Follow-the-Perturbed-Whittle-Leader has low regret} 
In this step, we consider the regret of Follow-the-Perturbed-Whittle-Leader (FPWL) described in Algorithm \ref{alg:FTWL}. In epoch $k$, a scheduling policy is chosen as follows:
\begin{equation}
\pi_k^{\text{FPWL}} = \text{Whittle}\bigg( \text{Perturb}\bigg(\sum_{t=1}^{k-1} f_t^{(1)},...,\sum_{t=1}^{k-1} f_t^{(N)}\bigg) \bigg),
\label{eq:fpwl_def}
\end{equation}
Unlike BWL and BPWL, this is a valid online learning algorithm in the full feedback setting since it does not require the cost functions in the current epoch and only uses past information. Now, we will bound the gap between the performance of FPWL and BPWL.

To do this, we state the following Lemma from \cite{kalai2005efficient}.

%\begin{framed}
	\begin{lemma}
		For any $v \in \mathbb{R}^{n}$, the cubes $\big[0,\frac{1}{\epsilon} \big]^{n}$ and $\big[0,\frac{1}{\epsilon} \big]^{n} + v$ overlap in at least a $(1-\epsilon |v|_1)$ fraction.
		\label{lem:overlap}
	\end{lemma}
%\end{framed}

We define the increment function $f'(\cdot)$ for an AoI cost function $f:\{1,...,M\} \rightarrow \mathbb{R}^{+}$ as follows:
\begin{equation}
	f'(i) = f(i) - f(i-1), \forall i \in 1,...,M.
\end{equation}
$f(0)$ is set to zero to have a valid definition for $i=1$. Now, we can rewrite the $\text{Perturb}(\cdot)$
using increment functions rather than cost functions. Thus,
\begin{equation}
	\text{Perturb}\bigg(f'^{(1)},...,f'^{(N)}\bigg) = \bigg(f'^{(1)} + \delta^{(1)},...,f'^{(N)}  + \delta^{(N)} \bigg),
\end{equation}
where $\delta^{(i)}(j) \sim \text{Uniform}\big([0,1/\epsilon]\big), \forall i \in 1,...,N \text{ and }\forall j \in 1,...,M$ are i.i.d. random variables. This allows us to write the perturbation procedure as an addition of i.i.d. uniform random vectors $\delta^{(i)}$. The earlier definition had $\gamma^{(i)}$ which were not element-wise i.i.d.

Now applying Lemma \ref{lem:overlap} we observe that \\ $\text{Perturb}\bigg(\sum\limits_{t=1}^{k-1} f_t^{(1)},...,\sum\limits_{t=1}^{k-1} f_t^{(N)}\bigg)$ and $\text{Perturb}\bigg(\sum\limits_{t=1}^{k} f_t^{(1)},...,\sum\limits_{t=1}^{k} f_t^{(N)}\bigg)$ have the same expectation with probability \[ \geq (1-\epsilon \sum_{i=1}^{N} |f_k'^{(i)}|_1).\] Using linearity of expectation and cost functions, $\mathbb{E}[C_k(\pi_k^{\text{FPWL}})]$ and $\mathbb{E}[C_k(\pi_k^{\text{BPWL}})]$ are also the same with probability \[ \geq (1-\epsilon \sum_{i=1}^{N} |f_k'^{(i)}|_1).\] On the non-overlapping fraction we assume the worst possible cost difference between the algorithms, which can be upper bounded by $D$ since we assume that the AoI cost functions are upper bounded by $D$. Combining all of this together, we get: 
\begin{equation}
	\mathbb{E}[C_k(\pi_k^{\text{FPWL}})] \leq \mathbb{E}[C_k(\pi_k^{\text{BPWL}})] + D\epsilon  \max_{f_k^{(1)},...,f_k^{(N)}} \sum_{i=1}^{N}  |f_k'^{(i)}|_1.
\end{equation}
Observe that since $f_k^{(i)}(\cdot) \leq D$, so $|f_k'^{(i)}|_1 \leq D$, for all $i$. Thus,

\begin{equation}
\mathbb{E}[C_k(\pi_k^{\text{FPWL}})] \leq \mathbb{E}[C_k(\pi_k^{\text{BPWL}})] + \epsilon ND^2.
\end{equation}

Adding the above equation for all $k \in 1,...,T$:

\begin{equation}
	\sum_{k=1}^{T} \mathbb{E}[C_k(\pi_k^{\text{FPWL}})] \leq  \sum_{k=1}^{T} \mathbb{E}[C_k(\pi_k^{\text{BPWL}})] + \epsilon N D^2 T.
\end{equation}
Using \eqref{eq:bpwl}, we finally have regret of the FPWL algorithm:

\begin{equation}
\mathbb{E}\big[\text{Regret}_T(\text{FPWL})\big] \leq \epsilon N D^2 T +  2\frac{M}{\epsilon} + \alpha T.
\label{eq:fpwl_regret}
\end{equation}

Setting $\epsilon = \sqrt{\frac{2M}{N D^2 T}}$, we get:

\begin{equation}
	\mathbb{E}\big[\text{Regret}_T(\text{FPWL})\big] \leq \alpha T + 2D\sqrt{2MNT}.
\end{equation}
This completes our proof.

\section{Proof of Lemma \ref{lem:dyn_mult}}
\label{pf:dyn_mult}
For any given sequence of cost functions $C_1,...,C_T$ that satisfy \eqref{eq:mult_vt}, the performance gap between the decisions $\pi_k$ given by \eqref{eq:fdwl} and choosing the optimal policy in each epoch is given by:
\begin{equation}
\sum_{k=1}^{T} C_k(\pi_k) - \sum_{k=1}^{T} \min_{ \pi }  C_k(\pi)
\end{equation}
We rewrite this as:
\begin{multline}
\sum_{k=2}^{T} \bigg(C_k\big( \argmin_{\pi \in \Pi} C_{k-1}(\pi) \big) - C_{k-1}\big( \argmin_{\pi \in \Pi} C_{k-1}(\pi) \big) \bigg) + \\C_1(\pi_1) - \min_{ x }  C_T(\pi) + \sum_{k=2}^{T} \bigg( C_k(\pi_k) - C_k\big( \argmin_{\pi \in \Pi} C_{k-1}(\pi) \big)  \bigg)
\end{multline}
Using Assumption \ref{ass:whittle_close} and the definition of $\pi_k$, we get:
\begin{equation}
	C_k(\pi_k) - C_k\big( \argmin_{\pi \in \Pi} C_{k-1}(\pi) \big) \leq \alpha, \forall k = 2,...,T.
	\label{eq:dynm1}
\end{equation}
This is because \[\argmin\limits_{\pi \in \Pi} C_{k-1}(\pi) = \text{Opt}(f_{k-1}^{(1)},...,f_{k-1}^{(N)}),\] while $\pi_k = \text{Whittle}(f_{k-1}^{(1)},...,f_{k-1}^{(N)})$.

Using the defintion of $V_T$ \eqref{eq:mult_vt}, the fact that $C_k(\cdot) \in [0,D]$ and the inequality \eqref{eq:dynm1}, we get:
\begin{equation}
\sum_{k=1}^{T} C_k(\pi_k) - \sum_{k=1}^{T} \min_{ \pi }  C_k(\pi) \leq \alpha T + V_T + D.
\end{equation}
Since the above equation is true for any sequence of cost functions that satisfy the $V_T$ constraint, it completes the proof.

\end{document}